\documentclass[11pt]{article}

\usepackage{amsmath} 
\usepackage{amsthm} 
\usepackage{amssymb}	
\usepackage{graphicx} 
\usepackage{multicol} 
\usepackage{multirow}
\usepackage{color}
\usepackage[dvips,letterpaper,margin=1in,bottom=1in]{geometry}

\usepackage[utf8]{inputenc}
\usepackage[english]{babel}

\usepackage{diagbox}
\usepackage{mathtools}

\newtheorem{theorem}{Theorem}[section]

\newtheorem{lemma}[theorem]{Lemma}

\newtheorem{proposition}[theorem]{Proposition}
\newtheorem{definition}{Definition}[section]

\newtheorem{remark}{Remark}[section]

\DeclarePairedDelimiter\rbra{\lparen}{\rparen}
\DeclarePairedDelimiter\sbra{\lbrack}{\rbrack}
\DeclarePairedDelimiter\cbra{\{}{\}}
\DeclarePairedDelimiter\abs{\lvert}{\rvert}
\DeclarePairedDelimiter\Abs{\lVert}{\rVert}
\DeclarePairedDelimiter\ceil{\lceil}{\rceil}
\DeclarePairedDelimiter\floor{\lfloor}{\rfloor}
\DeclarePairedDelimiter\ket{\lvert}{\rangle}
\DeclarePairedDelimiter\bra{\langle}{\rvert}

\newcommand{\set}[2] {\left\{\, #1 \colon #2 \,\right\}}

\usepackage{latexsym}
\usepackage{CJK}

\usepackage{enumerate}

\usepackage{algorithm}
\usepackage{algpseudocode}

\usepackage{stmaryrd}
\usepackage{booktabs}

\usepackage{hyperref}
\newcommand{\footremember}[2]{%
    \footnote{#2}
    \newcounter{#1}
    \setcounter{#1}{\value{footnote}}%
}

\usepackage{cite}
\usepackage{bbm}

\usepackage{tablefootnote}
\usepackage{threeparttable}
\usepackage{scrextend}

\usepackage{tikz}
\usetikzlibrary{positioning}

\usepackage{qtree}

\newcommand{\qisheng}[1]{{\color{red} [\textit{Qisheng: } #1]}}

\begin{document}

\title{Quantum Algorithm for Lexicographically Minimal String Rotation}
\author{
    Qisheng Wang 
    \footremember{1}{Qisheng Wang is with the Graduate School of Mathematics, Nagoya University, Nagoya, Japan (e-mail: \url{QishengWang1994@gmail.com}). Part of the work was done when the author was at the Department of Computer Science and Technology, Tsinghua University, Beijing, China.}
    \and Mingsheng Ying \footremember{7}{Mingsheng Ying is with the State Key Laboratory of Computer Science, Institute of Software, Chinese Academy of Sciences, Beijing, China, and also with the Department of Computer Science and Technology, Tsinghua University, Beijing, China (e-mail: \url{yingms@ios.ac.cn}).}
}
\date{}

\maketitle

\begin{abstract}
    Lexicographically minimal string rotation (LMSR) is a problem to find the minimal one among all rotations of a string in the lexicographical order, which is widely used in equality checking of graphs, polygons, automata and chemical structures.
    In this paper, we propose an $O(n^{3/4})$ quantum query algorithm for LMSR. In particular, the algorithm has average-case query complexity $O(\sqrt n \log n)$, which is shown to be asymptotically optimal up to a polylogarithmic factor, compared to its $\Omega\left(\sqrt{n/\log n}\right)$ lower bound. Furthermore, we show that our quantum algorithm outperforms any (classical) randomized algorithms in both worst and average cases. As an application, it is used in benzenoid identification and disjoint-cycle automata minimization.

\end{abstract}


    \newpage

    \tableofcontents
    \newpage

\section{Introduction}
\subsection{Lexicographically Minimal String Rotation}
    Lexicographically Minimal String Rotation (LMSR) is the problem of finding the lexicographically smallest string among all possible cyclic rotations of a given input string \cite{Boo80}. It has been widely used in equality checking of graphs \cite{Col80}, polygons \cite{Ili89, Mae91}, automata \cite{Pup10} (and their minimizations \cite{Alm08}) and chemical structures \cite{Shi79}, and in generating de Bruijn sequences \cite{Saw16,DHS+18} (see also \cite{Gro03, Lin01}). Booth \cite{Boo80} first proposed an algorithm in linear time for LMSR based on the Knuth-Morris-Pratt string-matching algorithm \cite{Knu77}. Shiloach \cite{Shi81} later improved Booth's algorithm in terms of performance. A more efficient algorithm was developed by Duval \cite{Duv83} from a different point of view known as Lyndon Factorization. All these algorithms for LMSR are deterministic and have worst-case time complexity $\Theta(n)$. After that, several parallel algorithms for LMSR were developed. Apostolico, Iliopoulos and Paige \cite{Apo87} found an $O(\log n)$ time CRCW PRAM (Concurrent Read Concurrent Write Parallel Random-Access Machine) algorithm for LMSR using $O(n)$ processors, which was then improved by Iliopoulos and Smyth \cite{Ili92} to use only $O(n/\log n)$ processors.

The LMSR of a string $s$ can also be computed by finding the lexicographically minimal suffix of $ss\$$, i.e. the concatenation of two occurrences of $s$ and an end-marker $\$$, where $\$$ is considered as a character lexicographically larger than every character in $s$. The minimal suffix of a string can be found in linear time with the help of data structures known as suffix trees \cite{Wei73, Aho74, McC76} and suffix arrays \cite{Man90, Gon92, Kar06}, and alternatively by some specific algorithms \cite{Duv83, Cro92, Ryt03} based on Duval's algorithm \cite{Duv83} or the KMP algorithm \cite{Knu77}.

    \subsection{Quantum Algorithms for String Problems}

    Although a large number of new quantum algorithms have been found for various problems (e.g., \cite{Sho94, Gro96, BS06, Amb07, MSS07, HHL09, BS17}), only a few of them solve string problems.

    Pattern matching is a fundamental problem in stringology, where we are tasked with determining whether a pattern $P$ of length $m$ occurs in a text $T$ of length $n$. In classical computing, it is considered to be closely related to LMSR.
    The Knuth-Morris-Pratt algorithm \cite{Knu77} used in Booth's algorithm \cite{Boo80} for LMSR mentioned above is one of the first few algorithms for pattern matching, with time complexity $\Theta(n + m)$. Recently, several quantum algorithms have been developed for pattern matching; for example, Ramesh and Vinay \cite{Ram03} developed an $O\left(\sqrt n \log (n/m) \log m + \sqrt{m}\log^2 m\right)$ quantum pattern matching algorithm based on a useful technique for parallel pattern matching, namely deterministic sampling \cite{Vis90}, and Montanaro \cite{Mon17} proposed an average-case $O\left((n/m)^d 2^{O\left(d^{3/2}\sqrt{\log m}\right)}\right)$ quantum algorithm for $d$-dimensional pattern matching.
    However, it seems that these quantum algorithms for pattern matching cannot be directly generalized to solve LMSR.

    Additionally, quantum algorithms for reconstructing unknown strings with nonstandard queries have been proposed; for example, substring queries \cite{CILG+12} and wildcard queries \cite{AM14}.
    Recently, a quantum algorithm that approximates the edit distance within a constant factor was developed in \cite{BEG+18}. Soon after, Le Gall and Seddighin \cite{LGS22} studied quantum algorithms for several other string problems: longest common substring, longest palindrome substring, and Ulam distance.

    \subsection{Main Contributions of This Paper}

    A naive quantum algorithm for LMSR (see Definition \ref{def:lmsr} for its formal definition) is to find the LMSR of the string within $O(\sqrt{n})$ comparisons of rotations by quantum minimum finding \cite{Dur96, Ahu99} among all rotations. However, each comparison of two rotations in the lexicographical order costs $O(\sqrt{n})$ and is bounded-error. Combining the both, an $\tilde O(n)$ quantum algorithm for LMSR is obtained, which has no advantages compared to classical algorithms.

    In this paper, however, we find a more efficient quantum algorithm for LMSR. Formally, we have:

    \begin{theorem} [Quantum Algorithm for LMSR] \label{thm-main2}
        There is a bounded-error quantum query algorithm for LMSR, for which the worst-case query complexity is $O\left(n^{3/4}\right)$ and average-case query complexity is $O\left(\sqrt{n} \log n\right)$.
    \end{theorem}

 In the top-level design of this algorithm, we are required to find the minimal value of a function, which is given by a bounded-error quantum oracle. To resolve this issue, we develop an efficient error reduction for nested quantum algorithms (see Section \ref{sec:intro-error-reduction} for an outline). With this framework of nested quantum algorithms, we are able to solve problems in nested structures efficiently.
The high level illustrations of the algorithm for the worst and average cases are given in Section \ref{sec:intro-exclusion-rule} and Section \ref{sec:intro-string-sensitivity}, respectively.
 A detailed description of the algorithm is presented in Section \ref{sec:algo}.

 We assume access to a quantum-read/classical-write random access memory (QRAM) and define time complexity as the number of elementary two-qubit quantum gates, input queries and QRAM operations (see Section \ref{sec:qram} for more details).
 Our quantum algorithm uses only $O(\log^2 n)$ classical bits in QRAM and $O(\log n)$ ``actual'' qubits in the quantum computation. Thus, the time complexity of our quantum algorithms in this paper is just an $O\rbra{\log n}$ factor bigger than their query complexity (in both worst and average cases).

    In order to show a separation between classical and quantum algorithms for LMSR, we settle the classical and quantum lower bounds for LMSR in both the worst and average cases. Let $R(\operatorname{LMSR})$ and $R^\mathit{unif}(\operatorname{LMSR})$ be the worst-case and average-case (classical) randomized query complexities for LMSR, and let $Q(\operatorname{LMSR})$ and $Q^\mathit{unif}(\operatorname{LMSR})$ be their quantum counterparts. Then we have:

    \begin{theorem} [Classical and Quantum Lower Bounds for LMSR] \label{thm-main3}
    \begin{enumerate}
        \item For every bounded-error (classical) randomized algorithm for LMSR, it has worst-case query complexity $\Omega(n)$ and average-case query complexity $\Omega(n / \log n)$. That is, $R(\operatorname{LMSR}) = \Omega(n)$ and $R^{\mathit{unif}}(\operatorname{LMSR}) = \Omega(n / \log n)$.
        \item For every bounded-error quantum algorithm for LMSR, it has worst-case query complexity $\Omega\left(\sqrt{n}\right)$ and average-case query complexity $\Omega\left(\sqrt{n / \log n}\right)$. That is, $Q(\operatorname{LMSR}) = \Omega(\sqrt n)$ and $Q^{\mathit{unif}}(\operatorname{LMSR}) = \Omega\left(\sqrt{n / \log n}\right)$.
    \end{enumerate}
    \end{theorem}
    \begin{remark}
        It suffices to consider only bounded-error quantum algorithms for LMSR, as we can show that every exact (resp. zero-error) quantum algorithm for LMSR has worst-case query complexity $\Omega(n)$. This is achieved by reducing the search problem to LMSR (see Appendix \ref{appendix-a}), since the search problem is known to have worst-case query complexity $\Omega(n)$ for exact and zero-error quantum algorithms \cite{Bea01}.
    \end{remark}

     Theorem \ref{thm-main3} is proved in Section \ref{sec:lowerbounds}. Our main proof technique is to reduce a total Boolean function to LMSR and to find a lower bound of that Boolean function based on the notion of block sensitivity. The key observation is that the block sensitivity of that Boolean function is related to the string sensitivity of input (Lemma \ref{lemma-bs-lowerbound}, and see Section \ref{sec:intro-string-sensitivity} for more discussions).

      {\vskip 4pt}

    The results of Theorems \ref{thm-main2} and \ref{thm-main3} can be summarized as Table \ref{tab-bounds}.

    \begin{table}[t]
    \centering
    \begin{tabular}{|c|c|c|c|c|}
    \hline
    \multirow{2}{*}{} & \multicolumn{2}{c|}{Classical} & \multicolumn{2}{c|}{Quantum} \\ \cline{2-5}
                      & Lower bounds  & Upper bounds  & Lower bounds  & Upper bounds \\ \hline
    Worst-case        & $\Omega(n)$    & $O(n)$ \cite{Boo80, Shi81, Duv83}       & $\Omega(\sqrt{n})$        &  $O\left(n^{3/4}\right)$            \\ \hline
    Average-case      & $\Omega(n/\log n)$    & $O(n)$ \cite{Ili94, Bas05}             & $\Omega\left(\sqrt{n/\log n}\right)$        &   $O\left(\sqrt{n}\log n\right)$           \\ \hline
    \end{tabular}
    \caption{Classical (randomized) and quantum query complexities of LMSR. }
    \label{tab-bounds}
    \end{table}

    Note that
    \[
        \Omega\left(\sqrt {n / \log n}\right) \leq Q^\mathit{unif}(\operatorname{LMSR}) \leq O\left(\sqrt n \log n\right).
    \]
    Therefore, our quantum algorithm is asymptotically optimal in the average case up to a logarithmic factor. Moreover, a quantum separation from (classical) randomized computation in both the worst-case and average-case query complexities is achieved:
    \begin{enumerate}
      \item Worst case: $Q(\operatorname{LMSR}) = O\left(n^{3/4} \right)$ but $R(\operatorname{LMSR}) = \Omega(n)$; and
      \item Average case: $Q^\mathit{unif}(\operatorname{LMSR}) = O\left(\sqrt n \log n\right)$ but $R^\mathit{unif}(\operatorname{LMSR}) = \Omega(n / \log n)$.
    \end{enumerate}
    In other words, our quantum algorithm is faster than any classical randomized algorithms both in the worst case and the average case.

     {\vskip 4pt}

    As an application, we show that our algorithm can be used in identifying benzenoids \cite{Bas16} and minimizing disjoint-cycle automata \cite{Alm08} (see Section \ref{sec:app}). The quantum speedups for these problems are illustrated in Table \ref{tab-algo}.

    \begin{table}[!hbp]
    \centering
    \begin{tabular}{|c|c|c|}
    \hline
              & Classical & Quantum (this work) \\ \hline
    LMSR      & $O(n)$ \cite{Boo80, Shi81, Duv83}   & $O(n^{3/4})$  \\ \hline
    Canonical boundary-edges code & $O(n)$ \cite{Bas16}   & $O(n^{3/4})$  \\ \hline
    Disjoint-cycle automata minimization & $O(mn)$ \cite{Alm08}   & $\tilde O(m^{2/3}n^{3/4})$  \\ \hline
    \end{tabular}
    \caption{Classical and quantum query complexities of LMSR, canonical boundary-edges code and disjoint-cycle automata minimization. For disjoint-cycle automata minimization, $m$ indicates the number of disjoint cycles and $n$ indicates the length of each cycle. Here, $\tilde O(\cdot)$ suppresses logarithmic factors.}
    \label{tab-algo}
    \end{table}

    \paragraph{Recent Developments.}

    After the work described in this paper, the worst-case quantum query complexity of LMSR was further improved to $n^{1/2+o(1)}$ in \cite{AJ22} by refining the exclusion rule of LMSR proposed in this paper (see Lemma 4.8 of \cite{AJ22}), and later a quasi-polylogarithmic improvement was achieved in \cite{Wan22}.
    A quantum algorithm for the decision version of LMSR with worst-case query complexity $\tilde O\rbra{\sqrt{n}}$ was proposed in \cite{Chi22} under their quantum divide-and-conquer framework.

    As an application, the near-optimal quantum algorithm for LMSR \cite{AJ22} was then used as a subroutine in finding the longest common substring of two input strings \cite{JN23}.

    \subsection{Overview of the Technical Ideas}

    Our results presented in the above subsection are achieved by introducing the following three new ideas:

    \subsubsection{Optimal Error Reduction for Nested Quantum Minimum Finding} \label{sec:intro-error-reduction}
    Our main algorithm for LMSR is essentially a nest of quantum search and minimum finding.
    A major difficulty in its design is error reduction in nested quantum oracles, which has not been considered in the previous studies of nested quantum algorithms (e.g., nested quantum search analyzed by Cerf, Grover and Williams \cite{Cer00} and nested quantum walks introduced by Jeffery, Kothari and Magniez \cite{Jef13}).

A $d$-level nested classical algorithm needs $O(\log^{d-1} n)$ repetitions to ensure a constant error probability by majority voting. For a $d$-level quantum algorithm composed of quantum minimum finding, it is known that only a small factor $O(\log n)$ of repetitions is required \cite{Cle08}. We show that this factor can be even better; that is, only $O(1)$ of repetitions are required as if there were no errors in quantum oracles:
    \begin{itemize}\item We extend quantum minimum finding algorithm \cite{Dur96, Ahu99} to the situation where the input is given by a bounded-error oracle so that it has query complexity $O(\sqrt{n})$ (see Lemma \ref{lemma-algo-min-find}) rather than $O(\sqrt{n} \log n)$ straightforwardly by majority voting.
  \item We introduce a success probability amplification method for quantum minimum finding on bounded-error oracles, which requires $O\left(\sqrt{n \log(1/\varepsilon)}\right)$ query complexity to obtain the minimum with error probability $\leq \varepsilon$ (see Lemma \ref{lemma-algo-min-find-amp}). In contrast, a straightforward solution by $O(\log(1/\varepsilon))$ repetitions of the $O(\sqrt{n})$ algorithm (by Lemma \ref{lemma-algo-min-find}) has query complexity $O(\sqrt{n} \log (1/\varepsilon))$.
   \end{itemize}
    These ideas are inspired by quantum searching on bounded-error oracles \cite{Hoy03} and amplification of success of quantum search \cite{Buh99}.
The above two algorithms will be used in the main algorithm for LMSR as a subroutine.
  Both of them are optimal because their simpler version (OR function) has lower bound $\Omega(\sqrt n)$ \cite{Ben97, Boy98, Zal99} for quantum bounded-error algorithm and $\Omega(\sqrt{n \log(1/\varepsilon)})$ \cite{Buh99} for its error reduction. To be clearer, we compare the results about quantum searching and quantum minimum finding in the previous literature and ours in Table \ref{tab-error-reduction}.

\begin{table}[!hbp]
\centering
\begin{tabular}{|c|c|c|c|}
\hline
Algorithm type                   & Oracle type   & Search             & Minimum finding                                               \\ \hline
\multirow{2}{*}{Bounded-error}   & exact         & $O(\sqrt{n})$ \cite{Gro96} & $O(\sqrt{n})$ \cite{Dur96, Ahu99} \\ \cline{2-4}
                                 & bounded-error & $O(\sqrt{n})$ \cite{Hoy03} & $O(\sqrt{n})$ (this work)                                                \\ \hline
\multirow{2}{*}{Error reduction} & exact         & $O\left(\sqrt{n \log(1/\varepsilon)}\right)$ \cite{Buh99}  & $O\left(\sqrt{n \log(1/\varepsilon)}\right)$ (this work)                                                \\ \cline{2-4}
                                 & bounded-error & $O\left(\sqrt{n \log(1/\varepsilon)}\right)$ \cite{Hoy03}  & $O\left(\sqrt{n \log(1/\varepsilon)}\right)$ (this work)                                                \\ \hline
\end{tabular}
\caption{Quantum query complexities of bounded-error quantum algorithms and their error reductions for search and minimum finding. }
\label{tab-error-reduction}
\end{table}

 Based on the above results, we develop an $O\left(\sqrt {n \log^{3} n \log \log n}\right)$ quantum algorithm for deterministic sampling \cite{Vis90}, and furthermore obtain an $O\left(\sqrt{n \log m} + \sqrt {m \log^{3} m \log \log m}\right)$ quantum algorithm for pattern matching, which is better than the best known result \cite{Ram03} of
 $O(\sqrt n \log (n/m) \log m + \sqrt m \log^2 m).
 $
 We also develop an optimal $O\left(\sqrt{n^d}\right)$ quantum algorithm for evaluating $d$-level shallow MIN-MAX trees that matches the lower bound $\Omega\left(\sqrt{n^d}\right)$ \cite{Amb00, Bar04} for AND-OR trees, and therefore it is optimal. The best known previous quantum query complexity of MIN-MAX trees is $O\left(W_d(n) \log n\right)$ \cite{Cle08}, where $W_d(n)$ is the query complexity of $d$-level AND-OR trees and optimally $O\left(\sqrt{n^d}\right)$ as known from \cite{Buh98, Hoy03}. Our improvements on these problems are summarized in Table \ref{tab-nested}.
    \begin{table}[!hbp]
    \centering
    \begin{tabular}{|c|c|c|}
    \hline
                           & Previous & Improved \\ \hline
    Deterministic sampling & $O\left(\sqrt{n}\log^2 n\right)$ \cite{Ram03}  & $O\left(\sqrt{n\log^{3}n\log \log n}\right)$         \\ \hline
    Pattern matching       & $O\left(\sqrt{n}\log(n/m)\log m\right)$ \cite{Ram03}        & $O\left(\sqrt{n\log m}\right)$         \\ \hline
    $d$-level MIN-MAX tree & $O\left(\sqrt{n^d} \log n\right)$ \cite{Cle08}        & $O\left(\sqrt{n^d}\right)$         \\ \hline
    \end{tabular}
    \caption{Quantum query complexities improved by our nested quantum algorithm framework.}
    \label{tab-nested}
    \end{table}

    \subsubsection{Exclusion Rule of LMSR} \label{sec:intro-exclusion-rule}
    We find a useful property of LMSR, named \textit{exclusion rule} (Lemma \ref{lemma-ex}): for any two overlapped substrings which are prefixes of the canonical representation of a string, the LMSR of the string cannot be the latter one. This property enables us to reduce the worst-case query complexity by splitting the string into blocks of suitable sizes, and in each block the exclusion rule can apply so that there is at most one candidate for LMSR. This kind of trick has been used in parallel algorithms, e.g., Lemma 1.1 of \cite{Ili92} and the Ricochet Property of \cite{Vis90}. However, the exclusion rule of LMSR used here is not found in the literature (to the best of our knowledge).

    We outline our algorithm as follows:
\begin{enumerate}
  \item[1] Let $B = \floor{\sqrt{n}}$ and $L = \floor{B/4}$. We split $s$ into $\ceil{n/L}$ blocks of size $L$ (except the last block).
  \item[2] Find the prefix $p$ of $\operatorname{SCR}(s)$ of length $B$, where $\operatorname{SCR}(s)$ is the canonical representation of $s$.
  \item[3] In each block, find the \textit{leftmost} index that matches $p$ as the candidate. Only the leftmost index is required because of the exclusion rule of LMSR.
  \item[4] Find the lexicographically minimal one among all candidates in blocks. In case of a tie, the minimal candidate is required.
\end{enumerate}
    A formal description and the analysis of this algorithm is given in Section \ref{sec:algo-basic}. In order to find the leftmost index that matches $p$ (Step 3 of this algorithm) efficiently, we adopt the deterministic sampling \cite{Vis90} trick. That is, we preprocess a deterministic sample of $p$, with which whether an index matches $p$ can be checked within $O(\log \abs{p})$ rather than $O(\abs{p})$ comparisons. Especially, we allow $p$ to be periodic, and therefore extend the definition of deterministic samples for periodic strings (see Definition \ref{def-deterministic-sampling}) and propose a quantum algorithm for finding a deterministic sample of a string (either periodic or aperiodic) (see Algorithm \ref{algo-deterministic-sampling}).

    \subsubsection{String Sensitivity} \label{sec:intro-string-sensitivity}
    We also observe the property of LMSR that almost all strings are of low string sensitivity (Lemma \ref{lemma-prob}), which can be used to reduce the query complexity of our quantum algorithm significantly in the average case. Here, the string sensitivity of a string (see Definition \ref{def-sensitivity}) is a metric showing the difficulty to distinguish its substrings, and helpful to obtain lower bounds for LMSR (see Lemma \ref{lemma-bs-lowerbound}).

    We outline our improvements for better average-case query complexity as follows:
    \begin{enumerate}
      \item[1] Let $s_1$ and $s_2$ be the minimal and the second minimal substrings of $s$ of length $B = O(\log n)$, respectively.
      \item[2] If $s_1 < s_2$ lexicographically, then return the starting index of $s_1$; otherwise, run the basic quantum algorithm given in Section \ref{sec:intro-exclusion-rule}.
    \end{enumerate}

    Intuitively, in the average case, we only need to consider the first $O(\log n)$ characters in order to compare two rotations. The correctness is straightforward but the average-case query complexity needs some analysis. See Section \ref{sec:algo-improved} for a formal description and the analysis of the improvements.

    \subsection{Organization of This Paper}
    We recall some basic definitions about strings and quantum query algorithms, and formally define the LMSR problem in Section \ref{sec:pre}. An efficient error reduction for nested quantum algorithms is developed in Section \ref{sec:qminimum}. An improved quantum algorithm for pattern matching based on the new error reduction technique (given in Section \ref{sec:qminimum}) is proposed in Section \ref{sec:qmatching}. The quantum algorithm for LMSR is proposed in Section \ref{sec:algo}. The classical and quantum lower bounds for LMSR are given in Section \ref{sec:lowerbounds}. The applications are discussed in Section \ref{sec:app}.

\section{Preliminaries} \label{sec:pre}

    For convenience of the reader, in this section, we briefly review the lexicographically minimal string rotation (LMSR) problem,
  quantum query model and several notions of worst-case and average-case complexities used in the paper.

   \subsection{Lexicographically Minimal String Rotation} \label{sec:lmsr-definition}

 For any positive number $n$, let $[n] = \{0, 1, 2, \dots, n-1\}$.
    Let $\Sigma$ be a finite alphabet with a total order $<$. A string $s \in \Sigma^n$ of length $n$ is a function $s: [n] \to \Sigma$. The empty string is denoted $\epsilon$. We use $s[i]$ to denote the $i$-th character of $s$. In case of $i \notin \mathbb{Z} \setminus [n]$, we define $s[i] \equiv s[i \bmod n]$. If $l \leq r$, $s[l \dots r] = s[l] s[l+1] \dots s[r]$ stands for the substring consisting of the $l$-th to the $r$-th character of $s$, and if $l > r$, we define $s[l \dots r] = \epsilon$. A prefix of $s$ is a string of the form $s[0 \dots i]$ for $i \in [n] \cup \{-1\}$. The period of a string $s \in \Sigma^n$ is the minimal positive integer $d$ such that $s[i] = s[i+d]$ for every $0 \leq i < n-d$. String $s$ is called \textit{periodic} if its period $\leq n/2$, and $s$ is \textit{aperiodic} if it is not periodic.

    Let $s \in \Sigma^n$ and $t \in \Sigma^m$. The concatenation of $s$ and $t$ is string $st = s[0] s[1] \dots s[n-1] t[0] t[1] \dots t[m-1]$. We write $s = t$ if $n = m$ and $s[i] = t[i]$ for every $i \in [n]$. We say that $s$ is smaller than $t$ in the \textit{lexicographical order}, denoted $s < t$, if either $s$ is a prefix of $t$ but $s \neq t$, or there exists an index $0 \leq k < \min\{n, m\}$ such that $s[i] = t[i]$ for $i \in [k]$ and $s[k] < t[k]$. For convenience, we write $s \leq t$ if $s < t$ or $s = t$.

    \begin{definition} [Lexicographically Minimal String Rotation] \label{def:lmsr}
    For any string $s \in \Sigma^n$ of length $n$, we call $s^{(k)} = s[k \dots k+n-1]$ the rotation of $s$ by offset $k$. The \textit{lexicographically minimal string rotation} (LMSR) problem is to find an offset $k$ such that $s^{(k)}$ is the minimal string among $s^{(0)}, s^{(1)}, \dots, s^{(n-1)}$ in the lexicographical order. The minimal $s^{(k)}$ is called the string canonical representation (SCR) of $s$, denoted $\operatorname{SCR}(s)$; that is,
    \[
        \operatorname{SCR}(s) = \min\left\{s^{(0)}, s^{(1)}, \dots, s^{(n-1)}\right\}.
    \]
    In case of a tie; that is, there are multiple offsets such that each of their corresponding strings equals to $\operatorname{SCR}(s)$, then the minimal offset is desired, and the goal is to find
    \[
        \operatorname{LMSR}(s) = \min\left\{ k \in [n]: s^{(k)} = \operatorname{SCR}(s) \right\}.
    \]
    \end{definition}

    The LMSR problem has been well-studied in the literature \cite{Col80, Boo80, Shi81, Duv83, Jeu93, Cro94, Cro07}, and several linear time (classical) algorithms for LMSR are known, namely Booth's, Shiloach's and Duval's Algorithms:
    \begin{theorem} [\cite{Boo80, Shi81, Duv83}] \label{thm-classical}
        There is an $O(n)$ deterministic algorithm for LMSR.
    \end{theorem}

    \subsection{Quantum Query Algorithms}

    Our computational model is the quantum query model \cite{Amb04, Buh02}. The goal is to compute an $n$-variable function $f(x)=f(x_0, x_1, \dots, x_{n-1})$, where $x_0, x_1, \dots, x_{n-1}$ are input variables. For example, the LMSR problem can be viewed as function $f(x) = \operatorname{LMSR}(x_0,x_1,\dots x_{n-1})$, where $x_0x_1x_2\dots x_{n-1}$ denotes the string of element $x_0, x_1, \dots, x_{n-1}$. The input variables $x_i$ can be accessed by queries to a quantum oracle $O_x$ (which is a quantum unitary operator) defined by $O_x \ket{i, j} = \ket{i, j \oplus x_i}$, where $\oplus$ is the bitwise exclusive-OR operation. A quantum algorithm $A$ with $T$ queries is described by a sequence of quantum unitary operators
    \[
        A: U_0 \to O_x \to U_1 \to O_x \to \dots \to O_x \to U_T.
    \]
    The intermediate operators $U_0, U_1, \dots, U_T$ can be arbitrary quantum unitary operators that are determined independent of $O_x$.
    The computation is performed in a Hilbert space $\mathcal{H} = \mathcal{H}_o \otimes \mathcal{H}_w$, where $\mathcal{H}_o$ is the output space and $\mathcal{H}_w$ is the work space. The computation starts from basis state $\ket 0_o \ket 0_w$, and then we apply $U_0, O_x, U_1, O_x, \dots, O_x, U_T$ on it in that order. The result state is
    \[
        \ket\psi = U_T O_x \dots O_x U_1 O_x U_0 \ket 0_o \ket 0_w.
    \]
    Measuring the output space, the outcome is then defined as the output $A(x)$ of algorithm $A$ on input $x$. More precisely, $\Pr[A(x) = y] = \Abs{M_y \ket\psi}^2$, where $M_y = \ket y_o \bra y$. Furthermore, $A$ is said to be a bounded-error quantum algorithm that computes $f$, if $\Pr[A(x) = f(x)] \geq 2/3$ for every $x$.

    To deal with average-case complexity, following the setting used in \cite{Amb99}, we assume that after each $U_j$, a dedicated flag-qubit will be measured on the computational basis (and this measurement may change the quantum state). The measurement outcome indicates whether the algorithm is ready to halt and return its output. If the outcome is $1$, then we measure the output space with the outcome as the output, and then stop the algorithm; otherwise, the algorithm continues with the next query $O_x$ and $U_{j+1}$.
    Let $T_A(x)$ denote the expected number of queries that $A$ uses on input $x$. Note that $T_A(x)$ only depends on the algorithm $A$ and its given input $x$ (which is fixed rather than from some distribution).

    \subsubsection{Worst-Case and Average-Case Query Complexities} \label{sec:qcomplexity}

    Let $f: \{0, 1\}^n \to \{0, 1\}$ be a Boolean function.
    If $A$ is a (either randomized or quantum) algorithm and $y \in \{0, 1\}$, we use $\Pr[A(x) = y]$ to denote the probability that $A$ outputs $y$ on input $x$. 
    Let $\mathcal{R}(f)$ and $\mathcal{Q}(f)$ be the set of randomized and quantum bounded-error algorithms that compute $f$, respectively:
    \begin{align*}
        \mathcal{R}(f) & = \{ \text{randomized algorithm }A: \forall x \in \{0, 1\}^n,\ \Pr[A(x) = f(x)] \geq 2/3 \}, \\
        \mathcal{Q}(f) & = \{ \text{quantum algorithm }A: \forall x \in \{0, 1\}^n,\ \Pr[A(x) = f(x)] \geq 2/3 \}.
    \end{align*}
    Then the \textit{worst-case} query complexities of $f$ are:
    \begin{align*}
        R(f) & = \inf_{A \in \mathcal{R}(f)} \max_{x \in \{0,1\}^n} T_A(x), \\
        Q(f) & = \inf_{A \in \mathcal{Q}(f)} \max_{x \in \{0,1\}^n} T_A(x).
    \end{align*}
    Let $\mu: \{0, 1\}^n \to [0, 1]$ be a probability distribution. We usually use $\mathit{unif} \equiv 2^{-n}$ to denote the uniform distribution. The \textit{average-case} query complexity of an algorithm $A$ with respect to $\mu$ is
    \[
        T_A^\mu = \mathbb{E}_{x \sim \mu} [T_A(x)] = \sum_{x \in \{0, 1\}^n} \mu(x) T_A(x).
    \]
    Thus, the randomized and quantum average-case query complexities of $f$ with respect to $\mu$ are:
    \begin{align*}
        R^\mu(f) & = \inf_{A \in \mathcal{R}(f)} T_A^\mu, \\
        Q^\mu(f) & = \inf_{A \in \mathcal{Q}(f)} T_A^\mu.
    \end{align*}
    Clearly, $Q^\mu(f) \leq R^\mu(f)$ for all $f$ and $\mu$.

    \subsubsection{Time and Space Efficiency} \label{sec:qram}

    In order to talk about the ``time'' and ``space'' complexities of quantum algorithms, we assume access to a quantum-read/classical-write random access memory (QRAM), where it takes a single QRAM operation to either classically write a bit to the QRAM or make a quantum query to a bit stored in the QRAM. For simplicity, we assume the access to the QRAM is described by a quantum unitary operator $U_{\text{QRAM}}$ that swaps the accumulator and a register indexed by another register:
    \[
        U_\text{QRAM} \ket{i, j} \ket{r_0, r_1, \dots, r_i, \dots, r_{M-1}} = \ket{i, r_i} \ket{r_0, r_1, \dots, j, \dots, r_{M-1}},
    \]
    where $r_0, r_1, \dots, r_{M-1}$ are bit registers that are only accessible through this QRAM operator.

    Let $A$ be a quantum query algorithm, and $t_A(x)$ denote the expected number of two-qubit quantum gates and QRAM operators $U_\text{QRAM}$ composing intermediate operators, and the quantum input oracles $O_x$ that $A$ uses on input $x$. The space complexity of $A$ measures the number of (qu)bits used in $A$. The worst-case and average-case time complexities of a Boolean function $f$ are defined similarly to Section \ref{sec:qcomplexity} by replacing $T_A(x)$ with $t_A(x)$.

    \section{Optimal Error Reduction for Nested Quantum Algorithms} \label{sec:qminimum}

   Our quantum algorithm for LMSR (Theorem \ref{thm-main2}) is essentially a nested algorithm calling quantum search and quantum minimum finding. The error reduction is often crucial for nested quantum algorithms. Traditional probability amplification methods for randomized algorithms can obtain an $O\left(\log^d n\right)$ slowdown for $d$-level nested quantum algorithms by repeating the algorithm $O(\log n)$ times in each level.
    In this section, we introduce an efficient error reduction for nested quantum algorithms composed of quantum search and quantum minimum finding, which only costs a factor of $O(1)$.
This improvement is obtained by finding an $O\left(\sqrt{n}\right)$ quantum algorithm for minimum finding when the comparison oracle can have bounded errors (see Algorithm \ref{algo-min-find}). Moreover, we also show how to amplify the success probability of quantum minimum finding with both exact and bounded-error oracles. In particular, we obtain an $O\left(\sqrt{n \log {(1/\varepsilon)}}\right)$ quantum algorithm for minimum finding with success probability $\geq 1 - \varepsilon$ (see Algorithm \ref{algo-min-find-amp}).
    These two algorithms allow us to control the error produced by nested quantum oracles better than traditional (classical) methods.
    Both of them are optimal because their simpler version (OR function) has lower bound $\Omega(\sqrt n)$ \cite{Ben97, Boy98, Zal99} for quantum bounded-error algorithm and $\Omega(\sqrt{n \log(1/\varepsilon)})$ \cite{Buh99} for its error reduction.
    As an application, we develop a useful tool to find the first solution in the search problem.


    \subsection{Quantum Search}

 Let us start from an $O\left(\sqrt{n}\right)$ quantum algorithm to search on bounded-error inputs \cite{Hoy03}. The search problem is described by a function $f(x_0, x_1, \dots, x_{n-1})$ that finds an index $j \in [n]$ (if exists) such that $x_j = 1$, where $x_i \in \{0, 1\}$ for all $i \in [n]$. It was first shown by Grover \cite{Gro96} that the search problem can be solved by an $O\left(\sqrt n\right)$ quantum algorithm, which
 was found after the discovery of the $\Omega\rbra{\sqrt{n}}$ lower bound \cite{Ben97} (see also \cite{Boy98, Zal99}).

    \subsubsection{Quantum Search on Bounded-Error Oracles}
    A more robust approach for the search problem on bounded-error oracles was proposed by H{\o}yer, Mosca and de Wolf \cite{Hoy03}. Rather than an exact quantum oracle $U_x \ket{i, 0} = \ket{i, x_i}$, they consider a bounded-error one introducing extra workspace $\ket{0}_w$:
    \[
        U_{x} \ket{i, 0} \ket{0}_w = \sqrt{p_i} \ket{i, x_i} \ket{\psi_{i}}_w + \sqrt{1 - p_i} \ket{i, \bar x_i} \ket{\phi_{i}}_w,
    \]
    where $p_i \geq 2/3$ for every $i \in [n]$, $\bar u$ denotes the negation of $u$, and $\ket{\psi_{i}}_w$ and $\ket{\phi_{i}}_w$ are ignorable work qubits. This kind of bounded-error oracles is general in the sense that every bounded-error quantum algorithm and (classical) randomized algorithm can be described by it. A naive way to solve the search problem on bounded-error oracles is to repeat $k = O(\log n)$ times and choose the majority value among the $k$ outputs. This gives an $O\left(\sqrt n \log n\right)$ quantum algorithm. Surprisingly, it can be made better to $O\left(\sqrt n\right)$ as shown in the following:

    \begin{theorem} [Quantum Search on Bounded-Error Oracles, \cite{Hoy03}] \label{thm-search-bounded-error-oracle}
        There is an $O\left(\sqrt n\right)$ bounded-error quantum algorithm for the search problem on bounded-error oracles. Moreover, if there are $t \geq 1$ solutions, the algorithm finds a solution in expected $O\left(\sqrt{n/t}\right)$ queries (even if $t$ is unknown).
    \end{theorem}

    For convenience, we use $\textbf{Search}(U_x)$ to denote the algorithm of Theorem \ref{thm-search-bounded-error-oracle} which, with probability $\geq 2/3$, returns an index $j \in [n]$ such that $x_j = 1$ or reports that no such $j$ exists (we require the algorithm to return $-1$ in this case).

    \subsubsection{Amplification of the Success of Quantum Search}

  Usually, we need to amplify the success probability of a quantum or (classical) randomized algorithm to make it sufficiently large. A common trick used in randomized algorithms is to repeat the bounded-error algorithm $O(\log (1/\varepsilon))$ times and choose the majority value among all outputs to ensure success probability $\geq 1 - \varepsilon$. Buhrman, Cleve, de Wolf and Zalka \cite{Buh99} showed that we can do better for quantum searching.

    \begin{theorem} [Amplification of the success of quantum search, \cite{Buh99}] \label{thm-amp-search}
        For every $\varepsilon > 0$, there is an $O\left(\sqrt{n \log (1/\varepsilon)}\right)$ bounded-error quantum algorithm for the search problem with success probability $\geq 1-\varepsilon$. Moreover, if there is a promise of $t \geq 1$ solutions, the algorithm finds a solution in $O\left(\sqrt n \left(\sqrt{t+\log(1/\varepsilon)} - \sqrt{t}\right)\right)$ queries.
    \end{theorem}

   Theorem \ref{thm-amp-search} also holds for bounded-error oracles. For convenience, we use $\textbf{Search}(U_x, \varepsilon)$ to denote the algorithm of Theorem \ref{thm-amp-search}, which succeeds with probability $\geq 1-\varepsilon$.
   Note that Theorem \ref{thm-amp-search} does not cover the case that there can be $t \geq 2$ solutions without promise. In this case, we can obtain an $O\left(\sqrt{n/t} \log(1/\varepsilon)\right)$ bounded-error quantum algorithm with error probability $\leq \varepsilon$ by straightforward majority voting.

    \subsection{Quantum Minimum Finding}

We now turn to consider the minimum-finding problem.
Given $x_0, x_1, \dots, x_{n-1}$, the problem is to find an index $j \in [n]$ such that $x_j$ is the minimal element. Let $\text{cmp}(i, j)$ be the comparator to determine whether $x_i < x_j$:
    \[
        \text{cmp}(i, j) = \begin{cases}
            1 & x_i < x_j,\\
            0 & \text{otherwise}.
        \end{cases}
    \]
    The comparison oracle $U_\text{cmp}$ simulating $\text{cmp}$ is defined by
    \[
        U_{\text{cmp}} \ket{i, j, k} = \ket{i, j, k \oplus \text{cmp}(i, j)}.
    \]
    We measure the query complexity by counting the number of queries to this oracle $U_\text{cmp}$. A quantum algorithm was proposed by D\"{u}rr and H{\o}yer \cite{Dur96} and Ahuja and Kapoor \cite{Ahu99}  for finding the minimum:

    \begin{theorem} [Minimum finding, \cite{Dur96, Ahu99}] \label{thm-find-min}
        There is an $O\left(\sqrt n\right)$ bounded-error quantum algorithm for the minimum-finding problem.
    \end{theorem}
    
    We also note that a generalized minimum-finding was developed in \cite{Ape17}, which only needs to prepare a superposition over the search space (rather than make queries to individual elements of the search space).

    \subsubsection{Optimal Quantum Minimum Finding on Bounded-Error Oracles}

 For the purpose of this paper, we need to generalize the above algorithm to one with a bounded-error version of $U_\text{cmp}$.
 For simplicity, we abuse a little bit of notation and define:
    \[
        U_\text{cmp} \ket{i, j, 0} \ket{0}_w = \sqrt{p_{ij}} \ket{i, j, \text{cmp}(i, j)} \ket{\psi_{ij}}_w + \sqrt{1 - p_{ij}} \ket{i, j, \overline{\text{cmp}(i, j)}} \ket{\phi_{ij}}_w,
    \]
    where $p_{ij} \geq 2/3$ for all $i, j \in [n]$, and $\ket{\psi_{ij}}_w$ and $\ket{\phi_{ij}}_w$ are ignorable work qubits. Moreover, for every index $j \in [n]$, we can obtain a bounded-error oracle $U_\text{cmp}^j$:
    \begin{equation} \label{eq:def-U_cmp^j}
        U_\text{cmp}^j \ket{i, 0} \ket{0}_w = \sqrt{p_{ij}} \ket{i, \text{cmp}(i, j)} \ket{\psi_{ij}}_w + \sqrt{1 - p_{ij}} \ket{i, \overline{\text{cmp}(i, j)}} \ket{\phi_{ij}}_w
    \end{equation}
    with only one query to $U_\text{cmp}$. Then we can provide a quantum algorithm for minimum finding on bounded-error oracles as Algorithm \ref{algo-min-find}.

    \begin{algorithm}
        \caption{$\textbf{Minimum}(U_\text{cmp})$: An algorithm for minimum finding on bounded-error oracles.}
        \label{algo-min-find}
        \begin{algorithmic}[1]
        \Require A bounded-error oracle $U_x$ for $x_0, x_1, \dots, x_{n-1}$.
        \Ensure An index $j \in [n]$ such that $x_j \leq x_i$ for every $i \in [n]$ with probability $\geq 2/3$.

        \If {$n \leq 2$}
            \State \Return the answer by classical algorithms.
        \EndIf
        \State $m \gets \ceil{12 \ln n}$, $q \gets \ceil{36 \ln m}$.
        \State Choose $j \in [n]$ uniformly at random.
        \For {$t = 1 \to m$}
            \If {the total number of queries to $U_\text{cmp} > 30C\sqrt{n} + mq$}
                \State \textbf{break}.
            \EndIf
            \State $i \gets \textbf{Search}(U_\text{cmp}^j)$, where $U_\text{cmp}^j$ is defined by Eq. (\ref{eq:def-U_cmp^j}).

            \For {$k = 1 \to q$}
            \State $b_k \gets \text{the measurement outcome of the third register of } U_{\text{cmp}}\ket{i,j,0}\ket{0}_w$.
            \EndFor

            \State $b \gets \rbra{i \neq -1} \land \operatorname{maj}\rbra{b_1, \dots, b_q}$, where $\operatorname{maj}\rbra{b_1, \dots, b_q}$ returns the majority of $b_1, \dots, b_q$.
            \If {$b$}
                \State $j \gets i$.
            \EndIf
        \EndFor
        \State \Return $j$.

        \end{algorithmic}
    \end{algorithm}

    The constant $C > 0$ in Algorithm \ref{algo-min-find} is given so that $\textbf{Search}(U_x)$ in Theorem \ref{thm-search-bounded-error-oracle} takes at most $C \sqrt{n/\max\{t, 1\}}$ queries to $U_x$ if there are $t$ solutions.

    \begin{lemma} \label{lemma-algo-min-find}
        Algorithm \ref{algo-min-find} is a bounded-error quantum algorithm for minimum finding on bounded-error oracles in $O\left(\sqrt n \right)$ queries.
    \end{lemma}

    \begin{proof}
        The query complexity is trivially $O(\sqrt{n})$ due to the guard (Line 7) of Algorithm \ref{algo-min-find}.

        The correctness is proved as follows. Let $m = \ceil{12 \ln n}$ and $q = \ceil{36 \ln m}$. We assume that $n \geq 3$ and therefore $m \geq 12$. In each of the $m$ iterations, Line 11-14 of Algorithm \ref{algo-min-find} calls $U_\text{cmp}$ for $q$ times and $b$ gets the value $\rbra{i \neq -1} \land \rbra{x_i < x_j}$ with probability $\geq 1-1/m^2$ (This is a straightforward majority voting. For completeness, its analysis is provided in Appendix \ref{appendix-maj-voting}).
        Here, $i$ is an index such that $x_i < x_j$ (with high probability) and $i = -1$ if no such $i$ exists; thus $b = 1$ means that there exists smaller element $x_i$ than $x_j$.

        We only consider the case that the values of $b$ in all iterations are as desired. This case happens with probability $\geq (1 - 1/m^2)^m \geq 1 - 1/m \geq 11/12$. In each iteration, $i$ finds a candidate index with probability $\geq 2/3$ such that $x_i < x_j$ if exists (and if there are many, any of them is obtained with equal probability). 
        It is shown in \cite[Lemma 2]{Dur96} that: if $i$ finds a candidate index with certainty, then the expected number of queries before $j$ holds the index of the minimal is $\leq \frac 5 2 C\sqrt{n}$; moreover, the expected number of iterations is $\leq \ln n$. In our case, $i$ finds a candidate index in expected $3/2$ iterations. Therefore, the expected number of queries to $U_\text{cmp}$ is $\leq \frac {15} 4 C\sqrt{n}$ and that of iterations is $\leq \frac 3 2 \ln n$. When Algorithm \ref{algo-min-find} makes queries to the oracle $\geq 30C\sqrt{n}$ times (except those negligible queries in Line 11-14) or iterations $\geq m$ times (that is, more than $8$ times their expectations), the error probability is $\leq 1/8 + 1/8 = 1/4$ by Markov's inequality. Therefore, the overall success probability is $\geq \frac {11} {12} \cdot \frac 3 4 \geq 2/3$.
    \end{proof}

    \subsubsection{Amplifying the Success Probability of Quantum Minimum Finding}

    We can amplify the success probability for quantum minimum finding better than a naive method, as shown in
    Algorithm \ref{algo-min-find-amp}.
 
    \begin{algorithm}
        \caption{$\textbf{Minimum}(U_\text{cmp}, \varepsilon)$: Amplification of the success of minimum finding.}
        \label{algo-min-find-amp}
        \begin{algorithmic}[1]
        \Require A bounded-error oracle $U_\text{cmp}$ for $x_0, x_1, \dots, x_{n-1}$ and $0 < \varepsilon < 1/2$.
        \Ensure An index $j \in [n]$ such that $x_j \leq x_i$ for every $i \in [n]$ with probability $\geq 1 - \varepsilon$.

        \While {\textbf{true}}
            \State $j \gets \textbf{Minimum}(U_\text{cmp})$.
            \If {$\textbf{Search}(U_\text{cmp}^j, \varepsilon) \neq -1$}
                \State \textbf{break}.
            \EndIf
        \EndWhile
        \State \Return $j$.

        \end{algorithmic}
    \end{algorithm}

    \begin{lemma} \label{lemma-algo-min-find-amp}
        Algorithm \ref{algo-min-find-amp} runs in expected $O\left(\sqrt{n \log {(1/\varepsilon)}}\right)$ queries with error probability $\leq \varepsilon$.
    \end{lemma}

    \begin{proof}
        Algorithm \ref{algo-min-find-amp} terminates with a guard by $\textbf{Search}(U_\text{cmp}^j, \varepsilon)$. Here, $\textbf{Search}(U_\text{cmp}^j, \varepsilon) \neq -1$ means that with probability $\geq 1-\varepsilon$, there is no index $i$ such that $\text{cmp}\rbra{i, j} = 1$ and thus $j$ is the desired answer. Therefore, it has error probability $\leq \varepsilon$ as the guard.
        Let $p \geq 2/3$ be the probability that $j$ holds the index of the minimal element with a single query to $\textbf{Minimum}(U_\text{cmp})$ by Lemma \ref{lemma-algo-min-find}. Let $q$ be the probability that Algorithm \ref{algo-min-find-amp} breaks the ``while'' loop at each iteration. Then
        \[
            q = p (1 - \varepsilon) + (1-p) \varepsilon \geq p(1-\varepsilon) \geq 1/3,
        \]
        which is greater than a constant. So, the expected number of iterations is $O(1)$. In a single iteration, $\textbf{Minimum}(U_\text{cmp})$ takes $O\left(\sqrt{n}\right)$ queries (by Lemma \ref{lemma-algo-min-find}) and $\textbf{Search}(U_\text{cmp}^j, \varepsilon)$ takes $O\left(\sqrt{n \log(1/\varepsilon)}\right)$ queries (by Theorem \ref{thm-amp-search}).
Therefore, the expected query complexity of Algorithm \ref{algo-min-find-amp} is $O\rbra{1} \cdot \rbra*{ O\rbra{\sqrt{n}}+O\left(\sqrt{n \log {(1/\varepsilon)}}\right) } = O\left(\sqrt{n \log {(1/\varepsilon)}}\right)$.
    \end{proof}

    \subsection{An Application: Searching for the First Solution}

In this subsection, we develop a tool needed in our quantum algorithm for LMSR as an application of the above two subsections. It solves the problem of
finding the first solution (i.e. leftmost solution, or solution with the minimal index rather than an arbitrary solution) and thus can be seen as
a generalization of quantum searching, but the solution is based on quantum minimum finding.

    Formally, the query oracle $U_x$ of $x_0, x_1, \dots, x_{n-1}$ is given. The searching-first problem is to find the minimal index $j \in [n]$ such that $x_j = 1$ or report that no solution exists.
    This problem can be solved by minimum-finding with the comparator
    \[
        \text{cmp}(i, j) = \begin{cases}
            1 & x_i = 1 \land x_j = 0 \\
            1 & x_i = x_j \land i < j \\
            0 & \text{otherwise}
        \end{cases},
    \]
    which immediately yields an $O(\sqrt n)$ solution if the query oracle $U_x$ is exact.

    In the case that the query oracle $U_x$ is bounded-error, a bounded-error comparison oracle $U_\text{cmp}$ corresponding to $\text{cmp}$ can be implemented with a constant number of queries to $U_x$. Therefore, the results in Lemma \ref{lemma-algo-min-find} and Lemma \ref{lemma-algo-min-find-amp} also hold for the searching-first problem. For convenience in the following discussions, we write $\textbf{SearchFirst}(U_x)$ and $\textbf{SearchFirst}(U_x, \varepsilon)$ to denote the algorithm for the searching-first problem based on the two algorithm $\textbf{Minimum}(U_\text{cmp})$ and $\textbf{Minimum}(U_\text{cmp}, \varepsilon)$, respectively. Symmetrically, we have $\textbf{SearchLast}(U_x)$ and $\textbf{SearchLast}(U_x, \varepsilon)$ for searching the last solution.

Recently, an $O(\sqrt{n})$ quantum algorithm for searching the first was proposed in \cite{Kap20}. Their approach is quite different from our presented above. It is specifically designed for this particular problem, but our approach is based on a more general framework of quantum minimum finding.

   {\vskip 4pt}

    \textit{We believe that the techniques presented in this section can be applied in solving other problems. For this reason, we present a description of them in a general framework of nested quantum algorithm in Appendix \ref{appendix-nested}. }

    \section{Quantum Deterministic Sampling} \label{sec:qmatching}

    In this section, we prepare another tool to be used in our quantum algorithm for LMSR, namely an efficient quantum algorithm for deterministic sampling. It is based on our nested quantum algorithm composed of quantum search and quantum minimum finding given in the last section. Deterministic sampling is also a useful trick in parallel pattern matching \cite{Vis90}. We provide a simple quantum lexicographical comparator in Section \ref{sec:lex-cmp}, and a quantum algorithm for deterministic sampling in Section \ref{sec:deterministic-sample}. As an application, we obtain quantum algorithms for string periodicity and pattern matching in Section \ref{sec:deterministic-app}.

    \subsection{Lexicographical Comparator} \label{sec:lex-cmp}

    Suppose there are two strings $s, t \in \Sigma^n$ of length $n$ over a finite alphabet $\Sigma = [\alpha]$. Let $U_s$ and $U_t$ be their query oracles, respectively. That is,
    \[
        U_s \ket{i, j} = \ket{i, j \oplus s[i]},\qquad\qquad  U_t \ket{i, j} = \ket{i, j \oplus t[i]}.
    \]
    In order to compare the two strings in the lexicographical order, we need to find the leftmost index $k \in [n]$ that $s[k] \neq t[k]$. If no such $k$ exists, then $s = t$. 
    To this end, we construct the oracle
    \begin{equation}
    \label{eq:def-Ux-comparator}
        U_x \ket{i, j} = \begin{cases}
            \ket{i, j}, & s\sbra{i} = t\sbra{i}, \\
            \ket{i, j \oplus 1}, & s\sbra{i} \neq t\sbra{i},
        \end{cases}
    \end{equation}
    using $1$ query to each of $U_s$ and $U_t$.
    A straightforward algorithm for lexicographical comparison based on the searching-first problem is described in Algorithm \ref{algo-lexcmp}.

    \begin{algorithm}
        \caption{$\textbf{LexicographicalComparator}(U_s, U_t)$: Lexicographical Comparator.}
        \label{algo-lexcmp}
        \begin{algorithmic}[1]
        \Require Two query oracles $U_s$ and $U_t$ for two strings $s$ and $t$, respectively.
        \Ensure If $s < t$, then return $-1$; if $s > t$, then return $1$; if $s = t$, then return $0$, with probability $\geq 2/3$.

        \State $k \gets \textbf{SearchFirst}(U_x)$, where $U_x$ is defined by Eq. (\ref{eq:def-Ux-comparator}).
        \If {$k = -1$}
            \State \Return $0$.
        \EndIf
        \If {$s[k] < t[k]$}
            \State \Return $-1$.
        \Else
            \State \Return $1$.
        \EndIf

        \end{algorithmic}
    \end{algorithm}

    \begin{lemma}
        Algorithm \ref{algo-lexcmp} is an $O\left(\sqrt n\right)$ bounded-error quantum algorithm that compares two strings by their oracles in the lexicographical order.
    \end{lemma}

    \begin{remark}
        We usually need to compare two strings in the lexicographical order as a subroutine nested as low-level quantum oracles in string algorithms. However, the lexicographical comparator (Algorithm \ref{algo-lexcmp}) brings errors. Therefore, the error reduction trick for nested quantum oracles proposed in Section \ref{sec:qminimum} is required here.
    \end{remark}

    \subsection{Deterministic Sampling} \label{sec:deterministic-sample}

    Deterministic sampling \cite{Vis90} is a useful technique for pattern matching in pattern analysis. In this subsection, we provide a quantum solution to deterministic sampling in the form of a nested quantum algorithm.

    For our purpose, we extend the definition of deterministic samples to the periodic case. The following is a generalised definition of deterministic samples.

    \begin{definition} [Deterministic samples] \label{def-deterministic-sampling}
    Let $s \in \Sigma^n$ and $d$ be its period. A \textit{deterministic sample} of $s$ consists of an offset $0 \leq \delta < \floor{n/2}$ and a sequence of indices $i_0, i_1, \dots, i_{l-1}$ (called \textit{checkpoints}) such that
    \begin{enumerate}
      \item $i_k-\delta \in [n]$ for $k \in [l]$;
      \item For every $0 \leq j < \floor{n/2}$ with $j \not\equiv \delta \pmod d$, there exists $k \in [l]$ such that $i_k-j \in [n]$ and $s[i_k - j] \neq s[i_k - \delta]$. We denote $c_k = s[i_k - \delta]$ when the exact values of $i_k$ and $\delta$ are ignorable.
    \end{enumerate}
    \end{definition}
    If $s$ is aperiodic (i.e. $d > n/2$), the second condition degenerates into ``for every $0 \leq j < \floor{n/2}$ with $j \neq \delta$'', which is consistent with the definition for aperiodic strings in \cite{Vis90}.

    \paragraph{The use of deterministic sampling.}
        Suppose that $T$ is a text and $P$ is a pattern. 
        If we have a deterministic sample $(\delta; i_0, i_1, \dots, i_{l-1})$ of $P$ with a small $l$, then we can test if an index of $T$ can be a starting position that matches $P$ using only $l$ comparisons according to the deterministic sample. 
        It is worth noting that one can disqualify two possible starting positions (that pass the above testing of deterministic sampling) by the Ricochet property proposed in \cite{Vis90} (see Lemma \ref{lemma-ricochet}).

    The following theorem shows that the size of the deterministic sample can be every small.

    \begin{theorem} [Deterministic sampling \cite{Vis90}]
        Let $s \in \Sigma^n$. There is a deterministic sample $(\delta; i_0, i_1, \dots, i_{l-1})$ of $s$ with $l \leq \floor{\log_2 n}$.
    \end{theorem}
    \begin{proof}

        For the case that $s$ is aperiodic, a simple procedure for constructing a valid deterministic sample was given in \cite{Vis90}. We describe it as follows.
        \begin{enumerate}
            \item Initially, let $A_0 = \sbra{\floor{n/2}}$ be the set of candidates of $\delta$, and $S_0 = \emptyset$ be the set of checkpoints. 
            \item At step $k \geq 0$, let $\delta_k^{\min} = \min A_k$ and $\delta_k^{\max} = \max A_k$. 
            \begin{enumerate}
                \item[2.1.] If $\delta_k^{\min} = \delta_k^{\max}$, then set $\delta = \delta_k^{\min}$ and return the current set $S_k$ of checkpoints. 
                \item[2.2.] Otherwise, there must be an index $i_k$ such that $s\sbra{i_k-\delta_k^{\min}} \neq s\sbra{i_k-\delta_k^{\max}}$ (or $s$ is periodic).
                Let $\sigma_k$ be the symbol with least occurrences (but at least once) among $s\sbra{i_k - \gamma}$ for $\gamma \in A_k$ (and choose any of them if there are multiple possible symbols). Let $S_{k+1} = S_k \cup \cbra{i_k}$ and $A_{k+1} = \set{\gamma \in A_k}{s\sbra{i_k-\gamma} = \sigma_k}$, then go to the next step for $k+1$.
            \end{enumerate}
        \end{enumerate}
        It can be seen that the above procedure always stops as the set $A_k$ halves after each step, i.e., $\abs{A_{k+1}} \leq \abs{A_k}/2$. 
        It can be verified that the returned $\delta$ and checkpoints together form a valid deterministic sample. 
        The procedure will have at most $\floor{\log_2 n}$ steps and each step will add one checkpoint, which implies that there exist a deterministic sample with at most $\floor{\log_2 n}$ checkpoints.

        For the case that $s$ is periodic, the above procedure will still work if we set the initial set of candidates to be $A_0 = \sbra{d}$. 
        Intuitively, since $s$ is periodic, most symbols are redundant and thus we only have to consider the first $d$ offsets. 
        After this modification, the analysis of the modified procedure is almost identical to the original one. 
        Here, we note that if $i_k$ does not exist during the execution of the modified procedure, then $s$ has a smaller period than $d$. 
        Finally, the procedure will return at most $\floor{\log_2 d} \leq \floor{\log_2 n}$ checkpoints. 
    \end{proof}

    Now let us consider how a quantum algorithm can do deterministic sampling. We start from the case where $s \in \Sigma^n$ is aperiodic. Let $U_s$ be the query oracle of $s$, that is, $U\ket{i, j} = \ket{i, j \oplus s[i]}$.
    Suppose at step $l$, the sequence of indices $i_0, i_1, \dots, i_{l-1}$ is known as well as $c_k = s[i_k-\delta]$ for $k \in [l]$ (we need not know $\delta$ explicitly). For $0 \leq j < \floor{n/2}$, let $x_j$ denote whether candidate $j$ agrees with $\delta$ at all checkpoints, that is,
    \begin{equation} \label{eq-deter-1}
        x_j = \begin{cases}
            0 & \exists k \in [l], j \leq i_k < j+n \land s[i_k - j] \neq c_k \\
            1 & \text{otherwise}
        \end{cases}.
    \end{equation}
    Based on the search problem, there is a bounded-error oracle $U_x$ for computing $x_j$ with $O\left(\sqrt l\right) = O\left(\sqrt{\log n}\right)$ queries to $U_s$.

    A quantum algorithm for deterministic sampling is described in Algorithm \ref{algo-deterministic-sampling}. Initially, all offsets $0 \leq \xi < \floor{n/2}$ are candidates of $\delta$. The idea of the algorithm is to repeatedly find two remaining candidates $p$ and $q$ that differ at an index $j$ (if there is only one remaining candidate, the algorithm has already found a deterministic sample and terminates), randomly choose a character $c$ being $s[j-p]$ or $s[j-q]$, and delete either $p$ or $q$ according to $c$. It is sufficient to select $p$ and $q$ to be the first and the last solution of $x_j$ defined in Eq. (\ref{eq-deter-1}).
    To explicitly describe how to find an index $j$ where $p$ and $q$ differ, we note that $q \leq j < p+n$ and $j$ must exist because of the aperiodicity of $s$, and let
    \[
        y_j = \begin{cases}
            1 & q \leq j < p+n \land s[j-p] \neq s[j-q] \\
            0 & \text{otherwise}
        \end{cases}.
    \]
    It is trivial that there is an exact oracle $U_y$ for computing $y_j$ with $O(1)$ queries to $U_s$.

For the case of periodic $s \in \Sigma^n$, the algorithm requires some careful modifications. We need a variable $Q$ to denote the upper bound of current available candidates. Initially, $Q = \floor{n/2} - 1$. We modify the definition of $x_j$ in Eq. (\ref{eq-deter-1}) to make sure $0 \leq j \leq Q$ by
    \begin{equation} \label{eq-deter-2}
        x_j = \begin{cases}
            0 & j > Q \lor \left( \exists k \in [l], j \leq i_k < j+n \land s[i_k - j] \neq c_k \right) \\
            1 & \text{otherwise}
        \end{cases}.
    \end{equation}
    For an aperiodic string $s$, there is at least one $y_j$ such that $y_j = 1$, so the algorithm will reach Line 17-25 during its execution with small probability $\leq 1/m$, where $m = O(\log n)$. But for periodic string $s$, let $d$ be its period, if $q-p$ is divisible by $d$, then $y_j = 0$ for all $y_j$ and thus the algorithm once reaches Line 17-25 with high probability $\geq 1-1/6m^2$. In this case, there does not exist $q \leq j < p+n$ such that $y_j = 1$. We set $Q = q-1$ to eliminate all candidates $\geq q$. In fact, even for a periodic string $s$, the algorithm is intended to reach Line 17-25 only once (with high probability). If the algorithm reaches Line 17-25 the second (or more) time, it is clear that current $(\delta; i_0, i_1, \dots, i_{l-1})$ is a deterministic sample of $s$ (with high probability $\geq 1-1/m$), and therefore consequent computation does not influence the correctness and can be ignored.

    \begin{algorithm}
        \caption{$\textbf{DeterministicSampling}(U_s)$: Deterministic Sampling.}
        \label{algo-deterministic-sampling}
        \begin{algorithmic}[1]
        \Require The query oracles $U_s$ for string $s \in \Sigma^n$.
        \Ensure A deterministic sample $(\delta; i_0, i_1, \dots, i_{l-1})$.

        \If {$n = 1$}
            \State \Return $(0; 0)$.
        \EndIf
        \State $m \gets \ceil{8 \log_2 n}$.
        \State $\varepsilon \gets 1/6m^2$.
        \State $l \gets 0$.
        \State $Q \gets \floor{n/2}-1$.
        \For {$t = 1 \to m$}
            \State Let $U_x$ be the bounded-error oracle for computing $x_j$ defined in Eq. (\ref{eq-deter-2}).
            \State $p \gets \textbf{SearchFirst}(U_x, \varepsilon)$.
            \State $q \gets \textbf{SearchLast}(U_x, \varepsilon)$.
            \If {$p = q$}
                \State \textbf{break}.
            \EndIf
            \State $i_l \gets \textbf{Search}(U_y, \varepsilon)$.
            \If {$i_l = -1$}
                \State $Q \gets q - 1$.
                \State $q \gets \textbf{SearchLast}(U_x, \varepsilon)$.
                \If {$p = q$}
                    \State \textbf{break}.
                \EndIf
                \State $i_l \gets \textbf{Search}(U_y, \varepsilon)$.
                \If {$i_l = -1$}
                    \State \textbf{break}.
                \EndIf
            \EndIf
            \State $c_l \gets s[i_l - p]$ with probability $1/2$ and $s[i_l - q]$ with probability $1/2$.
            \State $l \gets l + 1$.
        \EndFor

        \State Let $U_x$ be the bounded-error oracle for computing $x_j$ defined in Eq. (\ref{eq-deter-1}).
        \State $\delta \gets \textbf{SearchFirst}(U_x, \varepsilon)$.
        \State \Return $(\delta; i_0, i_1, \dots, i_{l-1})$.

        \end{algorithmic}
    \end{algorithm}

    \begin{lemma} \label{lemma-deterministic-sampling}
        Algorithm \ref{algo-deterministic-sampling} is an $O\left(\sqrt{n \log^3 n \log \log n}\right)$ bounded-error quantum algorithm for deterministic sampling.
    \end{lemma}
    \begin{proof}
        Assume $n \geq 2$ and let $m = \ceil{8 \log_2 n}$ and $\varepsilon = 1/6m^2$. There are $m$ iterations in Algorithm \ref{algo-deterministic-sampling}. In each iteration, there are less than $6$ calls to $\textbf{Search}$, $\textbf{SearchFirst}$ or $\textbf{SearchLast}$, which may bring errors. It is clear that each call to $\textbf{Search}$, $\textbf{SearchFirst}$ or $\textbf{SearchLast}$ has error probability $\leq \varepsilon$. Therefore, Algorithm \ref{algo-deterministic-sampling} runs with no errors from $\textbf{Search}$, $\textbf{SearchFirst}$ or $\textbf{SearchLast}$ with probability $\geq (1-\varepsilon)^{6m} \geq 1-1/m$.

        Now suppose Algorithm \ref{algo-deterministic-sampling} runs with no errors from $\textbf{Search}$, $\textbf{SearchFirst}$ or $\textbf{SearchLast}$. To prove the correctness of Algorithm \ref{algo-deterministic-sampling}, we consider the following two cases:
        \begin{enumerate}
          \item \textbf{Case 1.} $s$ is aperiodic. In this case, Algorithm \ref{algo-deterministic-sampling} will never reach Line 17-25. In each iteration, the leftmost and the rightmost remaining candidates $p$ and $q$ are found. If $p = q$, then only one candidate remains, and thus a deterministic sample is found. If $p \neq q$, then there exists an index $q \leq j < p+n$ such that $s[j-p]$ and $s[j-q]$ differ. We set $i_l = j$ and set $c_l$ randomly from $s[j-p]$ and $s[j-q]$ with equal probability. Then with probability $1/2$, half of the remaining candidates are eliminated. In other words, it is expected to find a deterministic sample in $2 \log_2 n$ iterations. The iteration limit to $m \geq 8 \log_2 n$ will make error probability $\leq 1/4$. That is, a deterministic sample is found with probability $\geq 3/4$.
          \item \textbf{Case 2.} $s$ is periodic and the period of $s$ is $d \leq n/2$. In each iteration, the same argument for aperiodic $s$ holds if Algorithm \ref{algo-deterministic-sampling} does not reach Line 17-25. If Line 17-25 is reached for the first time, it means candidates between $q + 1$ and $\floor{n/2}-1$ are eliminated and $q - p$ is divisible by $d$. If Line 17-25 is reached for the second time, all candidates $p \not\equiv q \pmod d$ are eliminated, and therefore a deterministic sample is found.
        \end{enumerate}
        Combining the above two cases, a deterministic sample is found with probability $\geq 3/4 (1 - 1/m) \geq 2/3$.

        On the other hand, we note that a single call to $\textbf{SearchFirst}$ and $\textbf{SearchLast}$ in Algorithm \ref{algo-deterministic-sampling} has $O\left(\sqrt{n \log (1/\varepsilon)}\right) = O\left(\sqrt{n \log \log n}\right)$ queries to $U_x$ (by Lemma \ref{lemma-algo-min-find-amp}), and $\textbf{Seach}$ has query complexity $O\left(\sqrt{n \log (1/\varepsilon)}\right) = O\left(\sqrt{n \log \log n}\right)$ (by Theorem \ref{thm-amp-search}). Hence, a single iteration has query complexity $O\left(\sqrt{nl\log \log n}\right) = O\left(\sqrt{n \log n \log \log n}\right)$, and the total query complexity ($m$ iterations) is $O\left(m \sqrt{n \log n \log \log n}\right) = O\left(\sqrt{n \log^3 n \log \log n}\right)$.
    \end{proof}

  Algorithm \ref{algo-deterministic-sampling} is a $2$-level nested quantum algorithm (see Appendix \ref{appendix-deterministic-sampling} for a more detailed discussion), and is a better solution for deterministic sampling in $O\left(\sqrt{n \log^3 n \log \log n}\right)$ queries than the known $O\left(\sqrt n \log^2 n\right)$ solution in \cite{Ram03}.

\begin{remark} \label{remark:time-eff}
  In order to make our quantum algorithm for deterministic sampling time and space efficient, we need to store and modify the current deterministic sample $(\delta; i_0, i_1, \dots, i_{l-1})$ during the execution in the QRAM, which needs $O(l \log n) = O(\log^2 n)$ bits of memory. Moreover, only $O(\log n)$ qubits are needed in the computation (for search and minimum finding). In this way, the time complexity of the quantum algorithm is $O\rbra*{\sqrt{n \log^5 n \log \log n}}$, which is just an $O\rbra{\log n}$ factor bigger than its query complexity.
\end{remark}

    \subsection{Applications} \label{sec:deterministic-app}

    Based on our quantum algorithm for deterministic sampling, we provide applications for string periodicity and pattern matching.

    \subsubsection{String Periodicity}

    We can check whether a string is periodic (and if yes, find its period) with its deterministic sample. Formally, let $(\delta; i_0, i_1, \dots, i_{l-1})$ be a deterministic sample of $s$, and $x_j$ defined in Eq. (\ref{eq-deter-1}). Let $j_1$ denote the smallest index $j$ such that $x_j = 1$, which can be computed by $\textbf{SearchFirst}$ on $x_j$. After $j_1$ is obtained, define
    \[
        x_j' = \begin{cases}
            0 & j \leq j_1 \lor \left( \exists k \in [l], j \leq i_k < j+n \land s[i_k - j] \neq c_k \right) \\
            1 & \text{otherwise}
        \end{cases}.
    \]
    Let $j_2$ denote the smallest index $j$ such that $x_j' = 1$ (if not found then $j_2 = -1$), which can be computed by $\textbf{SearchFirst}$ on $x_j'$.
    If $j_2 = -1$, then $s$ is aperiodic; otherwise, $s$ is periodic with period $d = j_2-j_1$. This algorithm for checking periodicity is $O\left(\sqrt{n \log n}\right)$. (See Appendix \ref{appendix-matching} for more details.)

    \subsubsection{Pattern Matching} \label{sec:pattern-matching}

    As an application of deterministic sampling, we have a quantum algorithm for pattern matching with query complexity $O\left(\sqrt{n \log m} + \sqrt {m \log^{3} m \log \log m}\right)$, better than the best known solution in \cite{Ram03} with query complexity $O\left(\sqrt n \log (n/m) \log m + \sqrt m \log^2 m\right)$.

   For readability, these algorithms are postponed to Appendix \ref{appendix-matching}.

    \section{The Quantum Algorithm for LMSR} \label{sec:algo}

    Now we are ready to present our quantum algorithm for LMSR and thus prove Theorem \ref{thm-main2}.
    This algorithm is designed in two steps:
    \begin{enumerate}
      \item Design a quantum algorithm with worst-case query complexity $O\left(n^{3/4} \right)$ in Section \ref{sec:algo-basic}; and
      \item Improve the algorithm to average-case query complexity $O\left(\sqrt{n} \log n\right)$ in Section \ref{sec:algo-improved}.
    \end{enumerate}

    \subsection{The Basic Algorithm} \label{sec:algo-basic}
    For convenience, we assume that the alphabet $\Sigma = [\alpha]$ for some $\alpha \geq 2$, where $[n] = \{ 0, 1, 2, \dots, n-1 \}$ and the total order of $\Sigma$ follows that of natural numbers.
    Suppose the input string $s \in \Sigma^n$ is given by an oracle $U_{\text{in}}$:
    \[
        U_{\text{in}} \ket{i, j} = \ket{i, j \oplus s[i]}.
    \]
The overall idea of our algorithm is to split $s$ into blocks of length $B$, and then in each block find a candidate with the help of the prefix of $\operatorname{SCR}(s)$ of length $B$. These candidates are eliminated between blocks by the exclusion rule for LMSR (see Lemma \ref{lemma-ex}). We describe it in detail in the next three subsections.

    \subsubsection{Find a Prefix of SCR}

    Our first goal is to find the prefix $p = s[\operatorname{LMSR}(s) \dots \operatorname{LMSR}(s)+B-1]$ of $\operatorname{SCR}(s)$ of length $B$ by finding an index $i^* \in [n]$ such that $s[i^* \dots i^*+B-1]$ matches $p$, where $B = \floor{\sqrt{n}}$ is chosen optimally (see later discussions). To achieve this, we need to compare two substrings of $s$ of length $B$ with the comparator $\text{cmp}_B$:
    \begin{equation} \label{eq-scr-cmpb}
        \text{cmp}_B(i, j) = \begin{cases}
            1 & s[i\dots i+B-1] < s[j\dots j+B-1] \\
            0 & \text{otherwise}
        \end{cases}.
    \end{equation}
    According to Algorithm \ref{algo-lexcmp}, we can obtain a bounded-error comparison oracle $U_{\text{cmp}_B}$ corresponding to $\text{cmp}_B$ with $O\left(\sqrt B\right)$ queries to $U_{\text{in}}$. After that, let $i^* \in [n]$ be any index such that $s[i^* \dots i^* + B - 1] = p$ by calling $\textbf{Minimum}(U_{\text{cmp}_B})$, which needs $O\left(\sqrt n \right)$ queries to $U_{\text{cmp}_B}$ (by Lemma \ref{lemma-algo-min-find}) and succeeds with a constant probability. In the following discussion, we use $i^*$ to find possible candidates of $\operatorname{LMSR}(s)$ and then find the solution among all candidates.

    \subsubsection{Candidate in Each Block}

    Being able to access contents of $p$ by $i^*$, we can obtain a deterministic sample of $p$ by Algorithm \ref{algo-deterministic-sampling} in $O\left(\sqrt {B \log^{3} B \log \log B}\right) = \tilde O\rbra{\sqrt{B}}$ queries with a constant probability.
    Suppose a deterministic sample of $p$ is known to be $(\delta; i_0, i_1, \dots, i_{l-1})$. We split $s$ into blocks of length $L = \floor{B/4}$. In the $i$-th block ($0$-indexed, $0 \leq i < \ceil{n/L}$), with the index ranging from $iL$ to $\min\{(i+1)L, n\} - 1$, a candidate $h_i$ is computed by
    \begin{equation} \label{eq-scr-candidate}
        h_i = \min \{ iL \leq j < \min\{(i+1)L, n\}: s[j \dots j+B-1] = p \},
    \end{equation}
    where the minimum is taken over all indices $j$ in the $i$-th block such that $s[j \dots j+B-1] = p$, and $\min \emptyset = \infty$.
    Intuitively, for each $0 \leq i < \ceil{n/L}$, $h_i$  defined by Eq. (\ref{eq-scr-candidate}) denotes the leftmost possible candidate for $\operatorname{LMSR}(s)$ such that $s[h_i \dots h_i+B-1] = p$ in the $i$-th block. On the other hand, $h_i$ denotes the first occurrence of $p$ with starting index in the $i$-th block of $s$, and thus can be computed by a procedure in quantum pattern matching (see Appendix \ref{appendix-matching} for more details), which needs $O\left(\sqrt{B \log B} \right) = \tilde O\rbra{\sqrt{B}}$ queries to $U_\text{in}$ with the help of the deterministic sample of $p$. We write $U_h$ for the bounded-error oracle of $h_i$. Note that $U_h$ is a $2$-level nested quantum oracle.

    \subsubsection{Candidate Elimination between Blocks}

    If we know the values of $h_i$ for $0 \leq i < \ceil{n/L}$, with either $h_i$ being a candidate or $\infty$ (indicating non-existence), then we can find $\operatorname{LMSR}(s)$ among all $h_i$'s with the comparator
    \begin{equation} \label{eq-scr-cmp}
        \text{cmp}(i, j) = \begin{cases}
            1 & \text{cmp}_n(h_i, h_j) = 1, \\
            1 & \text{cmp}_n(h_i, h_j) = \text{cmp}_n(h_i, h_j) = 0 \land h_i < h_j, \\
            0 & \text{otherwise,}
        \end{cases}
    \end{equation}
    where $\text{cmp}_n$ is defined by Eq. (\ref{eq-scr-cmpb}), and $\infty$ can be regarded as $n$ explicitly in the computation.
    Then we can obtain the bounded-error comparison oracle $U_\text{cmp}$ corresponding to $\text{cmp}$ with constant number of queries to $U_{\text{cmp}_n}$ and $U_h$, with $O\left(\sqrt{B \log B} + \sqrt n\right)$ queries to $U_\text{in}$. Here, $U_\text{cmp}$ is a $2$-level nested quantum oracle.
At the end of the algorithm, the value of $\operatorname{LMSR}(s)$ is chosen to be the minimal element among $h_i$ by comparison oracle $U_\text{cmp}$ according to comparator $\text{cmp}$. It can be seen that the algorithm is a $3$-level nested quantum algorithm.

    \subsubsection*{The Algorithm}

    We summarize the above design ideas in Algorithm \ref{algo-main}. There are four main steps (Line 5, Line 6, Line 7 and Line 8) in the algorithm. Especially, Line 7 of Algorithm \ref{algo-main} involves a $3$-level nested quantum algorithm. For convenience, we assume that each of these steps succeeds with a high enough constant probability, say $\geq 0.99$. To achieve this, each step just needs a constant number of repetitions to amplify the success probability from $2/3$ up to $0.99$.
    \begin{algorithm}
        \caption{$\textbf{BasicLMSR}(U_\text{in})$: Quantum algorithm for LMSR.}
        \label{algo-main}
        \begin{algorithmic}[1]
        \Require The query oracle $U_\text{in}$ for string $s$.
        \Ensure $\operatorname{LMSR}(s)$ with probability $\geq 2/3$.

        \If {$n \leq 15$}
            \State \Return $\operatorname{LMSR}(s)$ by classical algorithms in Theorem \ref{thm-classical}.
        \EndIf
        \State $B \gets \floor{\sqrt n}$.
        \State $i^* \gets \textbf{Minimum}(U_{\text{cmp}_B})$, where $\text{cmp}_B$ is defined by Eq. (\ref{eq-scr-cmpb}).
        \State $(\delta; i_0, i_1, \dots, i_{l-1}) \gets \textbf{DeterministicSampling}(U_{p})$, where $p = s[i^* \dots i^*+B-1]$.
        \State $i \gets \textbf{Minimum}(U_\text{cmp})$, where $\text{cmp}$ is defined by Eq. (\ref{eq-scr-cmp}).
        \State \Return $h_i$, where $h_i$ is defined by Eq. (\ref{eq-scr-candidate}).

        \end{algorithmic}
    \end{algorithm}

    \subsubsection*{Complexity}

    The query complexity of Algorithm \ref{algo-main} comes from the following four parts:
    \begin{enumerate}
      \item One call to $\textbf{Minimum}(U_{\text{cmp}_B})$, which needs $O\left(\sqrt n \right)$ queries to $U_{\text{cmp}_B}$ (by Lemma \ref{lemma-algo-min-find}), i.e. $O\left(\sqrt {nB}\right)$ queries to $U_\text{in}$.
      \item One call to $\textbf{DeterministicSampling}(U_{s[i^* \dots i^*+B-1]})$, which needs $O\left(\sqrt{B \log^{3} B \log \log B }\right)$ queries to $U_\text{in}$ (by Lemma \ref{lemma-deterministic-sampling}).
      \item One call to $\textbf{Minimum}(U_\text{cmp})$, which needs $O\left(\sqrt {n/L} \right)$ queries to $U_\text{cmp}$ (by Lemma \ref{lemma-algo-min-find}), i.e.
          \[
          O\left(\sqrt {n/L} \left(\sqrt{B \log B} + \sqrt n\right) \right) = O\left(n/\sqrt{B} \right)
          \]
          queries to $U_\text{in}$.
      \item Compute $h_i$, i.e. one query to $U_h$, which needs $O\left(\sqrt{B \log B} \right)$ queries to $U_\text{in}$.
    \end{enumerate}
    Therefore, the total query complexity is
    \[
        O\left(\sqrt{nB} + n/\sqrt{B} \right) = O\left(n^{3/4} \right)
    \]
    by selecting $B = \Theta\left(\sqrt{n}\right)$.

    \subsubsection*{Correctness}
    The correctness of Algorithm \ref{algo-main} is not obvious due to the fact that we only deal with one candidate in each block, but there can be several candidates that matches $p$ in a single block. This issue is resolved by the following exclusion rule:
    \begin{itemize}
      \item For every two equal substrings $s[i \dots i+B-1]$ and $s[j \dots j+B-1]$ of $s$ that overlap each other with $0 \leq i < j < n$ and $1 \leq B \leq n/2$, if both of them are prefixes of $\operatorname{SCR}(s)$, then $\operatorname{LMSR}(s)$ cannot be the larger index $j$.
    \end{itemize}
 More precisely, this exclusion rule can be stated as the following:
    \begin{lemma} [Exclusion Rule for LMSR] \label{lemma-ex}
        Suppose $s \in \Sigma^n$ is a string of length $n$. Let $2 \leq B \leq n/2$, and two indices $i, j \in [n]$ with $i < j < i+B$. If $s[i \dots i+B-1] = s[j \dots j+B-1] = s[\operatorname{LMSR}(s) \dots \operatorname{LMSR}(s)+B-1]$, then $\operatorname{LMSR}(s) \neq j$.
    \end{lemma}

    \begin{proof}
        See Appendix \ref{appendix-ex-scr}.
    \end{proof}

    Indeed, the above exclusion rule can be viewed as the Ricochet Property of LMSR. 
    Here, the Ricochet Property means that if two candidates are in the same block, then at most one of them can survive. This kind of Ricochet property was found to be useful in string matching, e.g., \cite{Vis90}.
    If there are two candidates in the same block, since each block is of length $L = \floor{B/4} < B$, then the two candidates must overlap each other. By this rule, the smaller candidate remains. Consequently, the correctness of Algorithm \ref{algo-main} is guaranteed because $h_i$ defined by Eq. (\ref{eq-scr-candidate}) always chooses the smallest candidate in each block.

    After the above discussions, we obtain:

     \begin{theorem} \label{thm-algo-main}
        Algorithm \ref{algo-main} is an $O\left(n^{3/4} \right)$ bounded-error quantum query algorithm for LMSR.
    \end{theorem}

    The quantum algorithm given by Theorem \ref{thm-algo-main} uses quantum deterministic sampling (Lemma \ref{lemma-deterministic-sampling}) as a subroutine. It can be also made time-efficient in the same way as discussed in Remark \ref{remark:time-eff}. 

    \subsection{An Improvement for Better Average-Case Query Complexity}\label{sec:algo-improved}

    In the previous subsection, we propose an efficient quantum algorithm for LMSR in terms of its worst-case query complexity. It is easy to see that its average-case query complexity remains the same as its worst-case query complexity.
    In this subsection, we give an improved algorithm with better average-case query complexity which also retains the worst-case query complexity.

    The basic idea is to individually deal with several special cases, which cover almost all of the possibilities on average. Let $B = \ceil{3 \log_\alpha n}$. Our strategy is to just consider substrings of length $B$. Let
    \[
    k = \arg \min_{i \in [n]} s[i \dots i+B-1]
    \]
    denote the index of the minimal substring among all substrings of length $B$, and then let
    \[
    k' = \arg \min_{i \in [n]\setminus\{k\}} s[i \dots i+B-1]
    \]
    denote the index of the second minimal substring among all substrings of length $B$. If $s[k \dots k+B-1] \neq s[k' \dots k'+B-1]$, then it immediately holds that $\operatorname{LMSR}(s) = k$. To find the second minimal substring, the index $k$ of the minimal substring should be excluded. For this, we need comparator $\text{cmp}_{B \setminus k}$:
    \begin{equation} \label{eq-scr-cmpbk}
        \text{cmp}_{B \setminus k}(i, j) = \begin{cases}
            0 & i = k \\
            1 & i \neq k \land j = k \\
            \text{cmp}_B(i, j) & \text{otherwise}
        \end{cases}.
    \end{equation}
    The bounded-error quantum comparison oracle $U_{\text{cmp}_{B \setminus k}}$ corresponding to $\text{cmp}_{B \setminus k}(i, j)$ can be defined with at most one query to $U_{\text{cmp}_B}$.

    \subsection*{The Algorithm}

    Our improved algorithm is presented as Algorithm \ref{algo-average}. It has three main steps (Line 5, Line 6 and Line 10). For the same reason as in Algorithm \ref{algo-main}, we assume that the third step (Line 10) succeeds with a high enough constant probability, say $\geq 0.99$.

    \begin{algorithm}
        \caption{$\textbf{ImprovedLMSR}(U_\text{in})$: Improved quantum algorithm for LMSR.}
        \label{algo-average}
        \begin{algorithmic}[1]
        \Require The query oracle $U_\text{in}$ for string $s$.
        \Ensure $\operatorname{LMSR}(s)$ with probability $\geq 2/3$.

        \If {$n \leq 3$}
            \State \Return $\operatorname{LMSR}(s)$ by classical algorithms.
        \EndIf
        \State $B \gets \ceil{3 \log_\alpha n}$.
        \State $k \gets \textbf{Minimum}(U_{\text{cmp}_B}, 1/2n)$, where $\text{cmp}_B$ is defined by Eq. (\ref{eq-scr-cmpb}).
        \State $k' \gets \textbf{Minimum}(U_{\text{cmp}_{B\setminus k}}, 1/2n)$, where $\text{cmp}_{B\setminus k}$ is defined by Eq. (\ref{eq-scr-cmpbk}).
        \If {$\text{cmp}_{B}(k, k') = 1$}
            \State \Return $k$.
        \Else
            \State \Return $\textbf{BasicLMSR}(U_\text{in})$.
        \EndIf

        \end{algorithmic}
    \end{algorithm}

    \subsection*{Correctness}

    The correctness of Algorithm \ref{algo-average} is trivial. We only consider the case where $n \geq 4$, and all of the three main steps succeed with probability
    \[
        \geq \left( 1 - \frac 1 {2n} \right)^2 \cdot 0.99 > \frac 2 3.
    \]
    If $\text{cmp}_{B}(k, k') = 1$, then $s[k \dots k+B-1]$ is the minimal substring of length $B$, and it immediately holds that $\operatorname{LMSR}(s) = k$. Otherwise, the correctness is based on that of $\textbf{BasicLMSR}(U_\text{in})$, which is guaranteed by Theorem \ref{thm-algo-main}.

    \subsection*{Complexity}

    The worst-case query complexity of Algorithm \ref{algo-average} is $O\left(n^{3/4} \right)$, obtained directly by Theorem \ref{thm-algo-main}. But settling the average-case query complexity of Algorithm \ref{algo-average} is a bit more subtle, which requires a better understanding of some properties of LMSR. To this end, we first introduce the notion of string sensitivity.

    \begin{definition} [String Sensitivity] \label{def-sensitivity}
        Let $s \in \Sigma^n$ be a string of length $n$ over a finite alphabet $\Sigma$. The string sensitivity of $s$, denoted $C(s)$, is the smallest positive number $l$ such that $s[i \dots i+l-1] \neq s[j \dots j+l-1]$ for all $0 \leq i < j < n$. In case that no such $l$ exists, define $C(s) = \infty$.
    \end{definition}

    The string sensitivity of a string is a metric indicating the difficulty to distinguish its rotations by their prefixes.
    If we know the string sensitivity $C(s)$ of a string $s$, we can compute $\operatorname{LMSR}(s)$ by finding the minimal string among all substrings  of $s$ of length $C(s)$, that is, $s[i \dots i+C(s)-1]$ for all $i \in [n]$.

    The following lemma shows that almost all strings have a low string sensitivity.
    \begin{lemma} [String Sensitivity Distribution] \label{lemma-prob}
        Let $s$ be a uniformly random string over $\Sigma^n$ and $1 \leq B \leq n/2$. Then
        \[
            \Pr[C(s) \leq B] \geq 1 - \frac 1 2 n(n-1) \alpha^{-B},
        \]
        where $\alpha = \abs{\Sigma}$. In particular,
        \[
            \Pr\left[C(s) \leq \ceil{3\log_\alpha n}\right] \geq 1 - \frac 1 n.
        \]
    \end{lemma}
    \begin{proof}
        Let $s \in \Sigma^n$ and $0 \leq i < j < n$. We claim that
        \[
        \Pr[s[i \dots i+B-1] = s[j \dots j+B-1]] = \alpha^{-B}.
        \]
        This can be seen as follows. Let $d$ be the number of members that appear in both sequences $\{ i \bmod n, (i+1) \bmod n, \dots, (i+B-1) \bmod n \}$ and $\{ j \bmod n, (j+1) \bmod n, \dots, (j+B-1) \bmod n \}$. It is clear that $0 \leq d < B$. We note that $s[i \dots i+B-1] = s[j \dots j+B-1]$ implies the following system of $B$ equations: $$s[i+k] = s[j+k]\ {\rm for\ all}\ 0 \leq k < B.$$ On the other hand, these $B$ equations involve $2B-d$ (random) characters. Therefore, there must be $(2B-d) - B = B-d$ independent characters, and the probability that the $B$ equations hold is
        \[
        \Pr[s[i \dots i+B-1] = s[j \dots j+B-1]] = \frac {\alpha^{B-d}} {\alpha^{2B-d}} = \alpha^{-B}.
        \]
        Consequently, we have:
        \begin{align*}
            \Pr[C(s) \leq B]
            & = 1 - \Pr[ \exists 0 \leq i < j < n, s[i \dots i+B-1] = s[j \dots j+B-1]] \\
            & \geq 1 - \sum_{i=0}^{n-1} \sum_{j = i+1}^{n-1} \Pr[s[i \dots i+B-1] = s[j \dots j+B-1]] \\
            & = 1 - \frac 1 2 n(n-1) \alpha^{-B}.
        \end{align*}
        In particular, in the case of $B = \ceil{3 \log_\alpha n}$, it holds that $\alpha^B \geq n^3$ and we obtain:
        \[
            \Pr[C(s) \leq B] \geq 1 - \frac 1 2 n(n-1) n^{-3} \geq 1 - \frac 1 n.
        \]
    \end{proof}

    With the above preparation, we now can analyze the average-case query complexity of Algorithm \ref{algo-average}. Let $s \in \Sigma^n$ be a uniformly random string over $\Sigma^n$ and $B = \ceil{3\log_\alpha n}$. Let $k$ and $k'$ denote the indices of the minimal and the second minimal substrings of length $B$ of $s$. To compute $k$ and $k'$, by Lemma \ref{lemma-algo-min-find-amp}, Algorithm \ref{algo-average} needs to make $O\left(\sqrt{n \log n}\right)$ queries to $U_{\text{cmp}_B}$, which is equivalent to $O\left(\sqrt{n \log n} \sqrt{B}\right) = O\left(\sqrt{n} \log n\right)$ queries to $U_\text{in}$. On the other hand, it requires $O(B) = O(\log n)$ queries to $U_\text{in}$ in order to check whether $\text{cmp}_{B}(k, k') = 1$, which is ignorable compared to other large complexities. Based on the result of $\text{cmp}_{B}(k, k')$, we only need to further consider the following two cases:

    \textbf{Case 1}. $\text{cmp}_{B}(k, k') = 1$. Note that this case happens with probability
    \begin{equation}\label{case1-prob}
        \Pr[\text{cmp}_{B}(k, k') = 1] \geq \Pr[C(s) \leq B] \geq 1 - \frac 1 n.
    \end{equation}
    In this case, Algorithm \ref{algo-average} returns $k$ immediately.

    \textbf{Case 2}. $\text{cmp}_{B}(k, k') \neq 1$. According to Eq. (\ref{case1-prob}), this case happens with probability $\leq 1/n$. In this case, Algorithm \ref{algo-average} makes one query to $\textbf{BasicLMSR}(U_\text{in})$, which needs $O\left(n^{3/4}\right)$ queries to $U_\text{in}$ (by Theorem \ref{thm-algo-main}).

  Combining the above two cases yields the average-case query complexity of Algorithm \ref{algo-average}:
    \[
        \leq O\left(\sqrt{n} \log n\right) + \left(1 - \frac 1 n\right) \cdot O(1) + \frac 1 n \cdot O\left(n^{3/4}\right) = O\left(\sqrt{n} \log n\right).
    \]

    After the above discussions, we obtain:
    \begin{theorem} \label{thm-algo-average}
        Algorithm \ref{algo-average} is an $O\left(n^{3/4} \right)$ bounded-error quantum query algorithm for LMSR, whose average-case query complexity is $O\left(\sqrt{n}\log n\right)$.
    \end{theorem}

Algorithm \ref{algo-average} for Theorem \ref{thm-algo-average} can be made time-efficient by an argument similar to that given in Section \ref{sec:algo-basic} for time and space efficiency. 

    \section{Lower Bounds of LMSR} \label{sec:lowerbounds}

    In this section, we establish average-case and worst-case lower bounds of both classical and quantum algorithms for the LMSR problem and thus prove Theorem \ref{thm-main3}.

    The notion of block sensitivity is the key tool we use to obtain lower bounds. Let $f: \{0, 1\}^n \to \{0, 1\}$ be a Boolean function. If $x \in \{0, 1\}^n$ is a binary string and $S \subseteq [n]$, we use $x^S$ to denote the binary string obtained by flipping the values of $x_i$ for $i \in S$, where $x_i$ is the $i$-th character of $x$:
    \[
        \left(x^S\right)_i = \begin{cases}
            \bar x_i & i \in S \\
            x_i & i \notin S
        \end{cases},
    \]
    where $\bar u$ denotes the negation of $u$, i.e. $\bar 0 = 1$ and $\bar 1 = 0$. The \textit{block sensitivity} of $f$ on input $x$, denoted $\mathit{bs}_x(f)$, is the maximal number $m$ such that there are $m$ disjoint sets $S_1, S_2, \dots, S_m \subseteq [n]$ for which $f(x) \neq f(x^{S_i})$ for $1 \leq i \leq m$.

    \subsection{Average-Case Lower Bounds}

    For settling the average-case lower bound, we need the following useful result about block sensitivities given in \cite{Amb99}.

    \begin{theorem} [General bounds for average-case complexity, {\cite[Theorem 6.3]{Amb99}}] \label{thm-amb}
        For every function $f: \{0, 1\}^n \to \{0, 1\}$ and probability distribution $\mu: \{0, 1\}^n \to [0, 1]$, we have $R^\mu(f) = \Omega\left(\mathbb{E}_{x \sim \mu}[\mathit{bs}_x(f)]\right)$ and $Q^\mu(f) = \Omega\left(\mathbb{E}_{x \sim \mu}\left[\sqrt{\mathit{bs}_x(f)}\right]\right)$.
    \end{theorem}

 In order to give a lower bound for LMSR by using Theorem \ref{thm-amb}, we need a binary function that can be reduced $\operatorname{LMSR}(x)$ to but is simpler than it. Here, we choose $\operatorname{LMSR}_0(x) = \operatorname{LMSR}(x) \bmod 2$. Obviously, we can compute $\operatorname{LMSR}(x)$, then $\operatorname{LMSR}_0(x)$ is immediately obtained. Moreover,  $\operatorname{LMSR}_0$ enjoys the following basic property:

    \begin{lemma} \label{lemma-bs-scr}
        Let $x \in \{0, 1\}^n$ and $0 \leq r < n$. Then $\mathit{bs}_x(\operatorname{LMSR}_0) = \mathit{bs}_{x^{(r)}}(\operatorname{LMSR}_0)$.
    \end{lemma}
    \begin{proof}
        Let $m = \mathit{bs}_x(\operatorname{LMSR}_0)$ and $S_1, S_2, \dots, S_m$ be the $m$ disjoint sets for which $\operatorname{LMSR}_0(x) \neq \operatorname{LMSR}_0\left(x^{S_i}\right)$ for $1 \leq i \leq m$. We define $S_i' = \{ (a+r) \bmod n: a \in S_i \}$ for $1 \leq i \leq m$. Then it can be verified that $\operatorname{LMSR}_0\left(x^{(r)}\right) \neq \operatorname{LMSR}_0\left(\left(x^{(r)}\right)^{S_i}\right)$ for $1 \leq i \leq m$. Hence, $\mathit{bs}_{x^{(r)}}(\operatorname{LMSR}_0) \geq m = \mathit{bs}_x(\operatorname{LMSR}_0)$.

        The same argument yields that $\mathit{bs}_{x}(\operatorname{LMSR}_0) \geq \mathit{bs}_{x^{(r)}}(\operatorname{LMSR}_0)$. Therefore, $\mathit{bs}_x(\operatorname{LMSR}_0) = \mathit{bs}_{x^{(r)}}(\operatorname{LMSR}_0)$.
    \end{proof}

    Next, we establish a lower bound for $\mathit{bs}_x(\operatorname{LMSR}_0)$.

    \begin{lemma} \label{lemma-bs-lowerbound}
        Let $x \in \{0, 1\}^n$. Then
        \begin{equation}\label{inequality-bs}
            \mathit{bs}_x(\operatorname{LMSR}_0) \geq \floor*{ \frac {n} {4C(x)} - \frac 1 4 },
        \end{equation}
        where $C(x)$ is the string sensitivity of $x$. 
    \end{lemma}
    \begin{proof} We first note that inequality (\ref{inequality-bs}) is trivially true when $C(x) > n/5$ because the right hand side is equal to $0$. For the case of $C(x) \leq n/5$, our proof is carried out in two steps:

    {\vskip 3pt}

    \textbf{Step 1}.
        Let us start from the special case of $\operatorname{LMSR}(x) = 0$. Note that $\operatorname{LMSR}_0(x) = 0$.
Let $B = C(x)$. We split $x$ into $(k + 1)$ substrings $x = y_1 y_2 \dots y_k y_{k+1}$, where $\abs{y_i} = B$ for $1 \leq i \leq k$, $\abs{y_{k+1}} = n \bmod B$, and $k = \floor{n/B}$. By the assumption that $C(x) = B$, we have $y_1 < \min\{y_2, y_3, \dots, y_k\}$. It holds that $y_1y_2 > 0^{2B}$; otherwise, $y_1 = y_2 = 0^B$, and then a contradiction $C(x) > B$ arises.

        Let $m = \floor{(k-1)/4}$. We select some of $y_i$s and divide them into $m$ groups (and the others are ignored). For every $1 \leq i \leq m$, the $i$-th group is $z_i = y_{4i-2} y_{4i-1} y_{4i} y_{4i+1}$. Let $L_i$ be the number of characters in front of $z_i$. Then $L_i = (4i-3)B$. We claim that $\mathit{bs}_x(\operatorname{LMSR}_0) \geq m$ by explicitly constructing $m$ disjoint sets $S_1, S_2, \dots, S_m \subseteq [n]$ such that $\operatorname{LMSR}_0(x) \neq \operatorname{LMSR}_0\left(x^{S_i}\right)$ for $1 \leq i \leq m$:
        \begin{enumerate}
          \item If $L_i$ is even, then we define:
          \[
            S_i = \{ L_i \leq j \leq L_i+2B: x_j \neq \delta_{j, L_i} \},
          \]
          where $\delta_{x, y}$ is the Kronecker delta, that is, $\delta_{x, y}=1$ if $x = y$ and $0$ otherwise.
          \item If $L_i$ is odd, then we define:
          \[
            S_i = \{ L_i+1 \leq j \leq L_i+2B+1: x_j \neq \delta_{j, L_i+1} \}.
          \]
        \end{enumerate}
        Note that $S_i \neq \emptyset$, and $0^{2B}$ is indeed the substring of $x^{S_i}$ that starts at the index $L_i+1$ if $L_i$ is even and at $L_i+2$ if $L_i$ is odd. That is,
        \[
            \operatorname{LMSR}\left(x^{S_i}\right) = \begin{cases}
                L_i+1 & L_i \equiv 0 \pmod 2 \\
                L_i+2 & L_i \equiv 1 \pmod 2
            \end{cases}.
        \]
        Then we conclude that $\operatorname{LMSR}_0\left(x^{S_i}\right) = 1$ for $1 \leq i \leq m$. Consequently,
        \[
        \mathit{bs}_x(\operatorname{LMSR}_0) \geq m = \floor*{\frac {k-1} 4} = \floor*{ \frac {n} {4B} - \frac 1 4 }.
        \]

        \textbf{Step 2}. Now we remove the condition that $\operatorname{LMSR}(x) = 0$ in Step 1. Let $r = \operatorname{LMSR}_0(x)$ and we consider the binary string $x^{(r)}$. Note that $\operatorname{LMSR}\left(x^{(r)}\right) = 0$. By Lemma \ref{lemma-bs-scr}, we have:
        \[
            \mathit{bs}_x(\operatorname{LMSR}_0) = \mathit{bs}_{x^{(r)}}(\operatorname{LMSR}_0) \geq \floor*{ \frac {n} {4B} - \frac 1 4 }.
        \]
        Therefore, inequality (\ref{inequality-bs}) holds for all $x \in \{0, 1\}^n$ with $C(x) \leq n/5$.
    \end{proof}

    We remark that inequality (\ref{inequality-bs}) can be slightly improved:
    \[
        \mathit{bs}_x(\operatorname{LMSR}_0) \geq \floor*{\frac{n-1}{2C(x)+2}},
    \]
    by splitting $x$ more carefully. However, Lemma \ref{lemma-bs-lowerbound} is sufficient for our purpose. With it, we obtain a lower bound of the expected value of $\mathit{bs}_x(\operatorname{LMSR}_0)$ when $x$ is uniformly distributed:

    \begin{lemma}
        Let $\mathit{unif}: \{0, 1\}^n \to [0, 1]$ be the uniform distribution that $\mathit{unif}(x) = 2^{-n}$ for every $x \in \{0, 1\}^n$. Then
        \begin{align*}
            \mathbb{E}_{x \sim \mathit{unif}}\left[\mathit{bs}_x(\operatorname{LMSR}_0)\right] & = \Omega\left(n/\log n\right), \\
            \mathbb{E}_{x \sim \mathit{unif}}\left[\sqrt{\mathit{bs}_x(\operatorname{LMSR}_0)}\right] & = \Omega\left(\sqrt{n/\log n}\right).
        \end{align*}
    \end{lemma}
    \begin{proof}
        By Lemma \ref{lemma-prob} and Lemma \ref{lemma-bs-lowerbound}, we have:
        \begin{align*}
            \mathbb{E}_{x \sim \mathit{unif}} [ \mathit{bs}_x(\operatorname{LMSR}_0) ]
            & = \sum_{x \in \{0,1\}^n} 2^{-n} \mathit{bs}_x(\operatorname{LMSR}_0) \\
            & \geq \sum_{x \in \{0,1\}^n} 2^{-n} \floor*{ \frac {n} {4C(x)} - \frac 1 4 } \\
            & = \sum_{B = 1}^n \sum_{x \in \{0,1\}^n: C(x) = B} 2^{-n} \floor*{ \frac {n} {4B} - \frac 1 4 } \\
            & \geq \sum_{B = 1}^{\ceil{3 \log_2 n}} \sum_{x \in \{0,1\}^n: C(x) = B} 2^{-n} \floor*{ \frac {n} {4B} - \frac 1 4 } \\
            & \geq \sum_{B = 1}^{\ceil{3 \log_2 n}} \sum_{x \in \{0,1\}^n: C(x) = B} 2^{-n} \floor*{ \frac {n} {4 \ceil{3 \log_2 n}} - \frac 1 4 } \\
            & = \floor*{ \frac {n} {4 \ceil{3 \log_2 n}} - \frac 1 4 } \Pr_{x \sim \mathit{unif}}[C(x) \leq \ceil{3 \log_2 n}] \\
            & \geq \floor*{ \frac {n} {4 \ceil{3 \log_2 n}} - \frac 1 4 } \left( 1 - \frac 1 n \right) \\
            & = \Omega\left(\frac{n}{\log n}\right).
        \end{align*}
        A similar argument yields that $\mathbb{E}_{x \sim \mathit{unif}}\left[\sqrt{\mathit{bs}_x(\operatorname{LMSR}_0)}\right] = \Omega\left(\sqrt{n/\log n}\right)$.
    \end{proof}

    By combining the above lemma with Theorem \ref{thm-amb}, we obtain lower bounds for randomized and quantum average-case bounded-error algorithms for LMSR:
    $$R^{\mathit{unif}}(\operatorname{LMSR}_0) = \Omega(n/\log n)\ {\rm and}\ Q^{\mathit{unif}}(\operatorname{LMSR}_0) = \Omega\left(\sqrt{n/\log n}\right).$$

    \subsection{Worst-Case Lower Bounds}\label{worst-case-section}

    Now we turn to consider the worst-case lower bounds. The idea is similar to the average case. First, the following result  similar to Theorem \ref{thm-amb} was also proved in \cite{Amb99}.

    \begin{theorem} [\hspace{1sp}\cite{Amb99}] \label{thm-amb-2}
        Let $A$ be a bounded-error algorithm for some function $f: \{0, 1\}^n \to \{0, 1\}$. \begin{enumerate}\item If $A$ is classical, then $T_A(x) = \Omega\left(\mathit{bs}_x(f)\right)$; and \item If $A$ is quantum, then $T_A(x) = \Omega\left(\sqrt{\mathit{bs}_x(f)}\right)$.\end{enumerate}
    \end{theorem}

    We still consider the function $\operatorname{LMSR}_0$ in this subsection. The following lemma shows that its block sensitivity can be linear in the worst case.

    \begin{lemma}
        There is a string $x \in \{0, 1\}^n$ such that $\mathit{bs}_x(\operatorname{LMSR}_0) \geq \floor{n/2}$.
    \end{lemma}
    \begin{proof}
        Let $x = 1^n$. Then $\operatorname{LMSR}(x) = 0$ and $\operatorname{LMSR}_0(x) = 0$. We can choose $m = \floor{n/2}$ disjoint sets $S_1, S_2, \dots, S_m$ with $S_i = \{ 2i-1 \}$. It may be easily verified that $\operatorname{LMSR}_0\left(x^{S_i}\right) = 1$ for every $1 \leq i \leq m$. Thus, by the definition of block sensitivity, we have $\mathit{bs}_x(\operatorname{LMSR}_0) \geq \floor{n/2}$.
    \end{proof}

    Combining the above lemma with Theorem \ref{thm-amb-2}, we conclude that $R(\operatorname{LMSR}_0) = \Omega(n)$ and $Q(\operatorname{LMSR}_0) = \Omega\left(\sqrt{n}\right)$, which give a lower bound for randomized and one for quantum worst-case bounded-error algorithms for LMSR, respectively. We have another more intuitive proof for the worst-case lower bound for quantum bounded-error algorithms and postpone into Appendix \ref{appendix-a}.

     \section{Applications} \label{sec:app}

    In this section, we present some practical applications of our quantum algorithm for LMSR.

    \subsection{Benzenoid Identification}

    The first application of our algorithm is a quantum solution to a problem about chemical graphs. Benzenoid hydrocarbons are a very important class of compounds \cite{Dia87, Dia88} and also popular as mimics of graphene (see \cite{Was08, Wan08, Pap16}).
    Several algorithmic solutions to the identification problem of benzenoids have been proposed in the previous literature; for example, Ba\v{s}i\'{c} \cite{Bas16} identifies benzenoids by boundary-edges code \cite{Han96} (see also \cite{Kov14}).

    Formally, the boundary-edges code (BEC) of a benzenoid is a finite string over a finite alphabet $\Sigma_6 = \{ 1, 2, 3, 4, 5, 6 \}$. The canonical BEC of a benzenoid is essentially the lexicographically maximal string among all rotations of any of its BECs and their reverses. Our quantum algorithm for LMSR can be used to find the canonical BEC of a benzenoid in $O\left(n^{3/4}\right)$ queries, where $n$ is the length of its BEC.

    More precisely, it is equivalent to find the lexicographically minimal one if we assume that the lexicographical order is $6 < 5 < 4 < 3 < 2 < 1$.
    Suppose a benzenoid has a BEC $s$. Our quantum algorithm is described as follows:
    \begin{enumerate}
      \item[1] Let $i = \operatorname{LMSR}(s)$ and $i^R = \operatorname{LMSR}(s^R)$, where $s^R$ denotes the reverse of $s$. This is achieved by Algorithm \ref{algo-main} in $O(n^{3/4})$ query complexity.
      \item[2] Return the smaller one between $s[i \dots i+n-1]$ and $s^R[i^R \dots i^R+n-1]$. This is achieved by Algorithm \ref{algo-lexcmp} in $O(\sqrt{n})$ query complexity.
    \end{enumerate}
    It is straightforward to see that the overall query complexity is $O(n^{3/4})$.

    \subsection{Disjoint-Cycle Automata Minimization}

    Another application of our algorithm is a quantum solution to minimization of a special class of automata. Automata minimization is an important problem in automata theory \cite{Hop00, Ber10} and has many applications in various areas of computer science. The best known algorithm for minimizing deterministic automata is $O(n \log n)$ \cite{Hop71}, where $n$ is the number of states. A few linear algorithms for minimizing some special automata are proposed in \cite{Rev92, Alm08}, which are important in practice, e.g., dictionaries in natural language processing.

    We consider the minimization problem of disjoint-cycle automata discussed by Almeida and Zeitoun in \cite{Alm08}. The key to this problem is a decision problem that checks whether there are two cycles that are equal to each other under rotations. Formally, suppose there are $m$ cycles, which are described by strings $s_1, s_2, \dots, s_m$ over a finite alphabet $\Sigma$. It is asked whether there are two strings $s_i$ and $s_j$ ($i \neq j$) such that $\operatorname{SCR}(s_i) = \operatorname{SCR}(s_j)$. For convenience, we assume that all strings are of equal length $n$, i.e. $\abs{s_1} = \abs{s_2} = \dots = \abs{s_m} = n$.

   A classical algorithm solving the above decision problem was developed in \cite{Alm08} with time complexity $O(mn)$. With the help of our quantum algorithm for LMSR, this problem can be solved more efficiently. We employ a quantum comparison oracle $U_\text{cmp}$ that compares strings by their canonical representations in the lexicographical order, where the corresponding classical comparator is:
    \[
        \text{cmp}(i, j) = \begin{cases}
            1 & \operatorname{SCR}(s_i) < \operatorname{SCR}(s_j), \\
            0 & \text{otherwise},
        \end{cases}
    \]
    and can be computed by finding $r_i = \operatorname{LMSR}(s_i)$ and $r_j = \operatorname{LMSR}(s_j)$. In particular, it can be done by our quantum algorithm for LMSR in $O\left(n^{3/4}\right)$ queries. Then the lexicographical comparator in Algorithm \ref{algo-lexcmp} can be use to compare $s_i[r_i \dots r_i + n - 1]$ and $s_j[r_j \dots r_j + n - 1]$ in query complexity $O\left(n^{3/4}\right)$. Furthermore, the problem of checking whether there are two strings that are equal to each other under rotations among the $m$ strings may be viewed as the element distinctness problem with quantum comparison oracle $U_\text{cmp}$, and thus can be solved by Ambainis's quantum algorithm \cite{Amb07} with $\tilde O\left(m^{2/3}\right)$ queries to $U_\text{cmp}$. In conclusion, the decision problem can be solved in quantum time complexity $\tilde O\left(m^{2/3}n^{3/4}\right)$, which is better than the best known classical $O(mn)$ time.

\section*{Acknowledgment}

We would like to thank the anonymous referees for their valuable comments and suggestions, which
helped to improve this paper.
Qisheng Wang would like to thank Fran{\c{c}}ois Le Gall for helpful discussions in an earlier version of this paper. 

This work was supported in part by the National Natural Science Foundation of China under Grant 61832015.
Qisheng Wang was also supported in part by the MEXT Quantum Leap Flagship Program (MEXT Q-LEAP) grants No.~JPMXS0120319794.

\addcontentsline{toc}{section}{References}

\bibliographystyle{alpha}
\bibliography{main}

\appendix
\section{Probability Analysis of Majority Voting} \label{appendix-maj-voting}

For convenience of the reader, let us first recall the Hoeffding's inequality \cite{Hoe63}.

\begin{theorem} [The Hoeffding's inequality, \cite{Hoe63}] \label{thm-hoeffding}
    Let $X_1, X_2, \dots, X_n$ be independent random variables with $0 \leq X_i \leq 1$ for every $1 \leq i \leq n$. We define the empirical mean of these variables by
    \[
        \bar X = \frac 1 n \sum_{i=1}^{n} X_i.
    \]
    Then for every $\varepsilon > 0$, we have:
    \[
        \Pr[\bar X \leq \mathbb{E}[\bar X] - \varepsilon] \leq \exp(-2n\varepsilon^2).
    \]
\end{theorem}

We only consider the case where $b$ is true.
In this case, we have $q = \ceil{36 \ln m}$ independent random variables $X_1, X_2, \dots, X_q$, where for $1 \leq i \leq q$, $X_i = 0$ with probability $\leq 1/3$ and $X_i = 1$ with probability $\geq 2/3$ (Here, $X_i = 1$ means the $i$-th query is true and $0$ otherwise). It holds that $\mathbb{E}[\bar X] \geq 2/3$. By Theorem \ref{thm-hoeffding} and letting $\varepsilon = 1/6$, we have:
\[
    \Pr\left[\bar X \leq \frac 1 2\right] \leq \Pr[\bar X \leq \mathbb{E}[\bar X] - \varepsilon] \leq \exp(-2q\varepsilon^2) \leq \exp(-2\ln m) = \frac 1 {m^2}.
\]
By choosing $\hat b$ to be true if $\bar X > 1/2$ and false otherwise, we obtain $\Pr\left[\hat b = b\right] = 1 - \Pr\left[\bar X \leq 1/2\right] \geq 1-1/m^2$ as required in Algorithm \ref{algo-min-find}.

\section{A Framework of Nested Quantum Algorithms}\label{appendix-nested}

    In this appendix, we provide a general framework to explain how the improvement given in Section \ref{sec:qminimum} can be achieved on nested quantum algorithms composed of quantum search and quantum minimum finding.

    This framework generalizes multi-level AND-OR trees. 
    In an ordinary AND-OR tree, a Boolean value is given at each of its leaves, and each non-leaf node is associated with an AND or OR operation which alternates in each level.
    A query-optimal quantum algorithm for evaluation of constant-depth AND-OR trees was proposed in \cite{Hoy03}. 
    Our framework can be seen as a generalization of fault tolerant quantum search studied in \cite{Hoy03}, which deals with the extended MIN-MAX-AND-OR trees which allow four basic operations and work for non-Boolean cases.
    Moreover, combining the error reduction idea in \cite{Buh99}, we also provide an error reduction for nested quantum algorithms. 
    Our results will be formally stated later in Lemma \ref{lemma-nested}.

    Suppose there is a $d$-level nested quantum algorithm composed of $d$ bounded-error quantum algorithms $A_1, A_2, \dots, A_d$ on a $d$-dimensional input $x(\theta_{d-1}, \dots, \theta_{2}, \theta_{1}, \theta_{0})$ given by an (exact) quantum oracle $U_0$:
    \[
        U_0 \ket{\theta_{d-1}, \dots, \theta_{2}, \theta_{1}, \theta_{0}, j} = U_0 \ket{\theta_{d-1}, \dots, \theta_{2}, \theta_{1}, \theta_{0}, j \oplus x(\theta_{d-1}, \dots, \theta_{2}, \theta_{1}, \theta_{0})},
    \]
    where $\theta_k \in [n_k]$ for $k \in [d]$ and $d \geq 2$.
For $1 \leq k \leq d$, $A_k$ is a bounded-error quantum algorithm given parameters $\theta_{d-1}, \dots, \theta_{k}$ that computes the function $f_k(\theta_{d-1}, \dots, \theta_k)$. In particular, $f_0 \equiv x$, and $A_d$ computes a single value $f_d$, which is considered to be the output of the nested quantum algorithm $A_1, A_2, \dots, A_d$.

Let $t_k \in \{s, m\}$ denote the type of $A_k$, where $s$ indicates the search problem and $m$ indicates the minimum finding problem. Let $S_m = \{ 1 \leq k \leq d: t_k = m \}$ denote the set of indices of the minimum finding algorithms. The behavior of algorithm $A_k$ is defined as follows:
    \begin{enumerate}
      \item \textbf{Case 1.} $t_k = s$. Then $A_k$ is associated with a checker $p_k(\theta_{d-1}, \dots, \theta_k, \xi)$ that determines whether $\xi$ is a solution (which returns $1$ for ``yes'' and $0$ for ``no''). $A_k$ is to find a solution $i \in [n_{k-1}]$ such that $$p_k(\theta_{d-1}, \dots, \theta_k, f_{k-1}(\theta_{d-1}, \dots, \theta_k, i)) = 1$$ (which returns $-1$ if no solution exists). Here, computing $p_k(\theta_{d-1}, \dots, \theta_k, f_{k-1}(\theta_{d-1}, \dots, \theta_k, i))$ requires the value of $f_{k-1}(\theta_{d-1}, \dots, \theta_k, i)$, which in turn requires a constant number of queries to $A_{k-1}$ with parameters $\theta_{d-1}, \dots, \theta_k, i$.
      \item \textbf{Case 2}. $t_k = m$. Then $A_k$ is associated with a comparator $c_k(\theta_{d-1}, \dots, \theta_k, \alpha, \beta)$ that compares $\alpha$ and $\beta$. $A_k$ is to find the index $i \in [n_{k-1}]$ of the minimal element such that
          $$c_k(\theta_{d-1}, \dots, \theta_k, f_{k-1}(\theta_{d-1}, \dots, \theta_{k}, j), f_{k-1}(\theta_{d-1}, \dots, \theta_{k}, i)) = 0$$
          for all $j \neq i$. Here, computing $c_k(\theta_{d-1}, \dots, \theta_k, f_{k-1}(\theta_{d-1}, \dots, \theta_{k}, j), f_{k-1}(\theta_{d-1}, \dots, \theta_{k}, i))$ requires the value of $f_{k-1}(\theta_{d-1}, \dots, \theta_{k}, j)$ and $f_{k-1}(\theta_{d-1}, \dots, \theta_{k}, i)$, which requires a constant number of queries to $A_{k-1}$ with parameters $\theta_{d-1}, \dots, \theta_k, j$ and $\theta_{d-1}, \dots, \theta_k, i$, respectively.
    \end{enumerate}

    The following lemma settles the query complexity of nested quantum algorithms, which shows that nested algorithms can do much better than naively expected.

    \begin{lemma} [Nested quantum algorithms] \label{lemma-nested}
        Given a $d$-level nested quantum algorithm  $A_1, A_2, \dots, A_d$, where $d \geq 2$, and $n_0, n_1, \dots, n_{d-1}$, $S_m, f_d$ are defined as above. Then:
        \begin{enumerate}\item There is a bounded-error quantum algorithm that computes $f_d$ with query complexity
        \[
            O\left( \sqrt{\prod_{k=0}^{d-1}  n_k} \right).
        \]
        \item There is a quantum algorithm that computes $f_d$ with error probability $\leq \varepsilon$ with query complexity
        \[
            O\left( \sqrt{\prod_{k=0}^{d-1}  n_k \log \frac 1 \varepsilon} \right).
        \]
        \end{enumerate}
    \end{lemma}

\begin{proof}
    Immediately yields by Theorem \ref{thm-search-bounded-error-oracle}, Theorem \ref{thm-amp-search}, Lemma \ref{lemma-algo-min-find} and Lemma \ref{lemma-algo-min-find-amp}.
\end{proof}

 \begin{remark}
        Lemma \ref{lemma-nested} can be seen as a combination of Theorem \ref{thm-search-bounded-error-oracle}, Theorem \ref{thm-amp-search}, Lemma \ref{lemma-algo-min-find} and Lemma \ref{lemma-algo-min-find-amp}. For convenience, we assume $n_0 = n_1 = \dots = n_{d-1} = n$ in the following discussions. Note that traditional probability amplification methods for randomized algorithms usually introduce an $O\left(\log^{d-1} n\right)$ slowdown for $d$-level nested quantum algorithms by repeating the algorithm $O(\log n)$ times in each level. However, our method obtains an extremely better $O(1)$ factor as if there were no errors in oracles at all.
    \end{remark}

    \begin{remark}
        Lemma \ref{lemma-nested} only covers a special case of nested quantum algorithms. A more general form of nested quantum algorithms can be described as a tree rather than a sequence, which allows intermediate quantum algorithms to compute their results by queries to several low-level quantum algorithms. We call them \textbf{adaptively nested quantum algorithms}, and an example of this kind algorithm is presented in Appendix \ref{appendix-matching} for pattern matching (see Figure \ref{fig-nested}).
    \end{remark}

\section{Remarks for Quantum Deterministic Sampling} \label{appendix-deterministic-sampling}

Algorithm \ref{algo-deterministic-sampling} uses several nested quantum algorithms as subroutines, but they are not described as nested quantum algorithms explicitly. Here, we provide an explicit description for Line 10 in Algorithm \ref{algo-deterministic-sampling} as an example. Let $\theta_0 \in [l]$ and $\theta_1 \in [n]$. Then:
\begin{itemize}\item The $0$-level function is
\[
    f_0(\theta_1, \theta_0) = \begin{cases}
        1 & \theta_1 \leq i_{\theta_0} < \theta_1 + n \land s[i_{\theta_0} - \theta_1] \neq s[i_{\theta_0} - \delta], \\
        0 & \text{otherwise},
    \end{cases}
\]
which checks whether candidate $\theta_1$ does not match the current deterministic sample at the $\theta_0$-th checkpoint. \item The $1$-level function is
\[
    f_1(\theta_1) = \begin{cases}
        0 & \theta_1 > Q \lor \exists i, f_0(\theta_1, i) = 1, \\
        1 & \text{otherwise},
    \end{cases}
\]
which checks candidate $\theta_1$ matches the current deterministic sample at all checkpoints. \item The $2$-level function is
\[
    f_2 = \min\left\{ i: f_1(i) = 1 \right\},
\]
which finds the first candidate of $\delta$. \end{itemize}

By Lemma \ref{lemma-nested}, $f_2$ can be computed with error probability $\leq \varepsilon = 1/6m^2$ in

     \[
    O\left(\sqrt{mn \log (1/\varepsilon)}\right) = O\left(\sqrt{n \log n \log \log n}\right).
\]

\section{Quantum Algorithm for Pattern Matching} \label{appendix-matching}

In this appendix, we give a detailed description of our quantum algorithm for pattern matching.

\subsection{Quantum Algorithm for String Periodicity}

The algorithm will be presented in the form of a nested quantum algorithm.
Suppose a string $s \in \Sigma^n$ is given by a quantum oracle $U_s$ that $U_s\ket{i, j} = \ket{i, j\oplus s[i]}$. We are asked to check whether $s$ is periodic, and if yes, find its period.

Let $(\delta; i_0, i_1, \dots, i_{l-1})$ be a deterministic sample of $s$.
We need a $2$-level nested quantum algorithm. Let $\theta_0 \in [l]$ and $\theta_1 \in [n]$. Then:\begin{itemize} \item The $0$-level function is
\[
    f_0(\theta_1, \theta_0) = \begin{cases}
        1 & \theta_1 \leq i_{\theta_0} < \theta_1+n \land s[i_{\theta_0}-\theta_1] = s[i_{\theta_0}-\delta], \\
        0 & \text{otherwise},
    \end{cases}
\]
which checks whether offset $\theta_1$ matches the deterministic sample at the $\theta_0$-th checkpoint. There is obviously an exact quantum oracle that computes $f_0(\theta_1, \theta_0)$ with a constant number of queries to $U_s$.
\item The $1$-level function is
\[
    f_1(\theta_1) = \begin{cases}
        0 & \exists i \in [l], f_0(\theta_1, i) = 1, \\
        1 & \text{otherwise},
    \end{cases}
\]
where $f_1(\theta_1)$ means offset $\theta_1$ matches the deterministic sample of $s$.
The $2$-level function is
\[
    f_2 = \min\{ i \in [n]: f_1(i) = 1 \},
\]
which finds the minimal offset that matches the deterministic sample. \end{itemize}

By Lemma \ref{lemma-nested}, we obtain an $O(\sqrt{nl})$ bounded-error quantum algorithm that finds the smallest possible offset $\delta_1$ of the deterministic sample of $s$. According to $\delta_1$, we define another $2$-level function
\[
    f_2' = \min\{ i \in [n]: i > \delta_1 \land f_1(i) = 1 \},
\]
which finds the second minimal offset that matches the deterministic sample of $s$, where $\min \emptyset = \infty$. Similarly, we can find the second smallest offset $\delta_2$ of the deterministic sample of $s$ in query complexity $O\left(\sqrt{nl} \right)$ with bounded-error. If $\delta_2 = \infty$, then $s$ is aperiodic; otherwise, $s$ is periodic with period $d = \delta_2 - \delta_1$. Therefore, we obtain an $O\left(\sqrt{n \log n}\right)$ bounded-error quantum algorithm that checks whether a string is periodic and, if yes, finds its period.

\subsection{Quantum Algorithm for Pattern Matching}

Suppose text $t \in \Sigma^n$ and pattern $p \in \Sigma^m$, and a deterministic sample of $p$ is $(\delta; i_0, i_1, \dots, i_{l-1})$.
    The idea for pattern matching is to split the text $t$ into blocks of length $L = \floor{m/4}$, the $i$-th ($0$-indexed) of which consists of indices ranged from $iL$ to $\min\{(i+1)L, n\} - 1$.
    Our algorithm applies to the case of $m \geq 4$, but does not to the case $1 \leq m \leq 3$, where a straightforward quantum search is required (we omit it here).

    The key step for pattern matching is to find a candidate $h_i$ for $0 \leq i < \ceil{n/L}$, indicating the first occurrence in the $i$-th block with starting index $iL \leq j < \min\{(i+1)L, n\}$, where $j+m \leq n$ and $t[j \dots j+m-1] = p$. Formally,
    \begin{equation} \label{eq-deter-3}
        h_i = \min\{ iL \leq j < \min\{(i+1)L, n\}: j+m \leq n \land t[j \dots j+m-1] = p \},
    \end{equation}
    where $\min \emptyset = \infty$.
    Note that Eq. (\ref{eq-scr-candidate}) is similar to Eq. (\ref{eq-deter-3}) but without the condition of $j + m \leq n$, which can easily removed from our algorithm given in the following discussion.
    In the previous subsection, we presented an efficient quantum algorithm that checks whether the string is periodic or not. Then we are able to design an quantum algorithm for pattern matching. Let us consider the cases of aperiodic patterns and periodic patterns, separately:

\subsubsection{Aperiodic Patterns}

    \textbf{The Algorithm}. For an aperiodic pattern $p$, $h_i$ can be computed in the following two steps:
    \begin{enumerate}
      \item Search $j$ from $iL \leq j < \min\{(i+1)L, n\}$ such that $t[i_k - \delta + j] = p[i_k - \delta]$ for every $k \in [l]$ (we call such $j$ a candidate of matching). For every $j$, there is an $O(\sqrt l)$ bounded-error quantum algorithm to check whether $t[i_k - \delta + j] = p[i_k - \delta]$ for every $k \in [l]$. Therefore, finding any $j$ requires $O\left(\sqrt {Ll}\right) = O\left(\sqrt{m \log m}\right)$ queries (by Theorem \ref{thm-search-bounded-error-oracle}).
      \item Check whether the index $j$ found in the previous step satisfies $j+m \leq n$ and $t[j \dots j+m-1] = p$. This can be computed by quantum search in $O\left(\sqrt m\right)$ queries. If the found index $j$ does not satisfy this condition, then $h_i = \infty$; otherwise, $h_i = j$.
    \end{enumerate}
    It is clear that $h_i$ can be computed with bounded-error in $O\left(\sqrt{m \log m}\right)$ queries according to the above discussion. Recall that $p$ appears in $t$ at least once, if and only if there is at least a $0 \leq i < \ceil{n/L}$ such that $h_i = 1$. By Theorem \ref{thm-search-bounded-error-oracle}, this can be checked in $O\left(\sqrt {n/L} \sqrt{m \log m}\right) = O\left(\sqrt {n \log m}\right)$ queries.

  {\vskip 3pt}

    \textbf{Correctness}. Note that in the $i$-th block, if there is no such $j$ or there are more than two values of $j$ such that $t[i_k - \delta + j] = p[i_k - \delta]$ for every $k \in [l]$, then there is no matching in the $i$-th block. More precisely, we have:
    \begin{lemma} [The Ricochet Property \cite{Vis90}] \label{lemma-ricochet}
        Let $p \in \Sigma^m$ be aperiodic, $(\delta; i_0, i_1, \dots, i_{l-1})$ be a deterministic sample, and $t \in \Sigma^n$ be a string of length $n \geq m$, and let $j \in [n-m+1]$. If $t[i_k+j-\delta] = p[i_k-\delta]$ for every $k \in [l]$, then $t[j' \dots j'+m-1] \neq p$ for every $j' \in [n-m+1]$ with $j-\delta \leq j' < j-\delta+\floor{m/2}$.
    \end{lemma}
    \begin{proof}
        Assume that $t[j' \dots j'+m-1] = p$ for some $j-\delta \leq j' < j-\delta+\floor{m/2}$. Let $x = \delta-j+j'$. Note that $0 \leq x < \floor{m/2}$. We further assume that $x \neq \delta$. Then by the definition of a deterministic sample, there exists $k \in [l]$ such that $x \leq i_k < x+m$ and $p[i_k-x] \neq p[i_k - \delta]$. On the other hand, $p[i_k-x] = t[j'+i_k-x] = t[i_k+j-\delta] = p[i_k-\delta]$. A contradiction arises, which implies $x = \delta$ and therefore $j = j'$.
    \end{proof}
    We can see that if there are two different indices $j_1, j_2$ in the $i$-th block such that $t[i_k - \delta + j_r] = p[i_k - \delta]$ for every $k \in [l]$ and $r \in \{1, 2\}$, it must hold that $\abs{j_1-j_2} \leq m/4$, since each block has length $L = \floor{m/4}$. If we apply Lemma \ref{lemma-ricochet} on $j_1$, then $t[j_2 \dots j_2+m-1] \neq p$; and if we apply it on $j_2$, then $t[j_1 \dots j_1+m-1] \neq p$. Thus, we conclude that neither $j_1$ nor $j_2$ can be a starting index of occurrence $p$ in $t$. As a result, there is at most one candidate of matching in each block.

    {\vskip 3pt}

    \textbf{A Description in the Form of a Nested Quantum Algorithm}. The above algorithm can be more clearly described as a $3$-level nested quantum algorithm. \begin{itemize}\item The level-$0$ function is defined by
    \[
        f_0(\theta_2, \theta_1, \theta_0) = \begin{cases}
            1 & i_{\theta_0} - \delta + \theta_2L + \theta_1 \in [n] \land t[i_{\theta_0} - \delta + \theta_2L + \theta_1] = p[i_{\theta_0} - \delta] \\
            0 & \text{otherwise}
        \end{cases},
    \]
    which checks whether $t$ matches $p$ at the $\theta_0$-th checkpoint at the $\theta_1$-th index in the $\theta_2$-th block, where $\theta_0 \in [l]$, $\theta_1 \in [L]$ and $\theta_2 \in \left[\ceil{n/L}\right]$.
    \item The level-$1$ function is defined by
    \[
        f_1(\theta_2, \theta_1) = \begin{cases}
            0 & \exists i \in [l], f_0(\theta_2, \theta_1, i) = 0, \\
            1 & \text{otherwise},
        \end{cases}
    \]
    which checks whether $t$ matches the deterministic sample at the $\theta_1$-th index in the $\theta_2$-th block.
    \item The level-$2$ function $f_2(\theta_2)$ finds a solution $i \in [L]$ such that $f_1(\theta_2, i) = 1$, which indicates the (only) candidate in the $\theta_2$-th block.
    \item The level-$3$ function $f_3$ finds a matching among $f_2(i)$ over all $i \in [\ceil{n/L}]$ by checking whether $iL + f_2(i) + m \leq n$ and $t[iL + f_2(i) \dots iL + f_2(i) + m - 1] = p$, where the latter condition can be checked by a quantum searching algorithm, which can be formulated as a 1-level nested quantum algorithm: \begin{itemize}\item The level-$0$ function
    \[
        g_0(\xi_2, \xi_1, \xi_0) = \begin{cases}
            1 & t[\xi_2L + \xi_1 + \xi_0] = p[\xi_0], \\
            0 & \text{otherwise},
        \end{cases}
    \]
    which checks whether the $\xi_0$-th character of substring of $t$ starting at offset $\xi_1$ in the $\xi_2$-th block matches the $\xi_0$-th character of $p$, where $\xi_0 \in [m]$, $\xi_1 \in [L]$ and $\xi_0 \in \left[\ceil{n/L}\right]$.
    \item The $1$-level function
    \[
        g_1(\xi_2, \xi_1) = \begin{cases}
            0 & \exists i \in [m], g_0(\xi_2, \xi_1, i) = 0,\\
            1 & \text{otherwise},
        \end{cases}
    \]
    which checks whether the substring of $t$ starting at offset $\xi_1$ in the $\xi_2$-th block matches $p$. Finally, we have that $f_3$ finds a solution $i \in \left[\ceil{n/L}\right]$ such that $iL + f_2(i) + m \leq n$ and $g_1(i, f_2(i)) = 1$.
    \end{itemize}
\end{itemize}

The structure of our algorithm can be visualized as the tree in Figure \ref{fig-nested}. It is worth noting that $f_3$ calls both $f_2$ and $g_1$.

    \begin{figure} [!hbp]
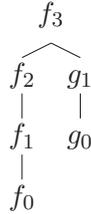
 
    \Tree [.$f_3$ [.$f_2$ [.$f_1$ $f_0$ ] ] [.$g_1$ [.$g_0$ ] ] ]
    \caption{Quantum pattern matching algorithm for aperiodic strings.}
    \label{fig-nested}
    \end{figure}

    By a careful analysis, we see that the query complexity of the above algorithm is $O\left(\sqrt{n \log m}\right)$.

\subsubsection{Periodic Patterns}

    For a periodic pattern $p$, a similar result can be achieved with some minor modifications to the algorithm for an aperiodic pattern. First, we have:
    \begin{lemma} [The Ricochet Property for periodic strings] \label{lemma-ricochet-period}
        Let $p \in \Sigma^m$ be periodic with period $d \leq m/2$, $(\delta; i_0, i_1, \dots, i_{l-1})$ be a deterministic sample, and $t \in \Sigma^n$ be a string of length $n \geq m$, and let $j \in [n-m+1]$. If $t[i_k+j-\delta] = p[i_k-\delta]$ for every $k \in [l]$, then $t[j' \dots j'+m-1] \neq p$ for every $j' \in [n-m+1]$ with $j-\delta \leq j' < j-\delta+\floor{m/2}$ and $j' \not\equiv j \pmod d$.
    \end{lemma}
    \begin{proof}
        Similar to proof of Lemma \ref{lemma-ricochet}.
    \end{proof}

 Now in order deal with periodic pattern $p$, the algorithm for an aperiodic pattern can be modified as follows. For the $i$-th block, in order to compute $h_i$, we need to find (by minimum finding) the leftmost and the rightmost candidates $j_l$ and $j_r$, which requires $O\left(\sqrt{Ll} \right) = O\left(\sqrt{m \log m} \right)$ queries (by Algorithm \ref{algo-min-find}). Let us consider two possible cases:
    \begin{enumerate}
      \item If $j_l \not\equiv j_r \pmod d$, then by Lemma \ref{lemma-ricochet-period}, there is no matching in the $i$-th block and thus $h_i = \infty$;
      \item If $j_l \equiv j_r \pmod d$, find the smallest $j_l \leq q \leq R_i$ such that $t[q \dots R_i] = p[q-j_l \dots q-j_l+R_i-q]$, where $R_i = \min\{(i+1)L, n\} - 1$ denotes the right endpoints of the $i$-th block, by minimum finding in $O\left(\sqrt L \right) = O\left(\sqrt m \right)$ queries (by Algorithm \ref{algo-min-find}). Then the leftmost candidate will be
          \[
          j = j_l + \ceil*{\frac {q-j_l} d} d.
          \]
          If $j \leq \min \left\{ n - m, j_r \right\}$ and $t[j \dots j+m-1] = p$ (which can be checked by quantum search in $O(\sqrt m)$ queries), then $h_i = j$; otherwise, $h_i = \infty$.
    \end{enumerate}

    \textbf{Correctness}. It is not straightforward to see that the leftmost occurrence in the $i$-th block is found in the case $j_l \equiv j_r \pmod d$ if there does exist an occurrence in that block. By Lemma \ref{lemma-ricochet-period}, if there exists an occurrence starting at index $j_l \leq j \leq j_r$ in the $i$-th block, then $j_l \equiv j \equiv j_r \pmod d$. Let the leftmost and the rightmost occurrences of $p$ in the $i$-th block be $j_l'$ and $j_r'$, respectively. Then $j_l \leq j_l' \leq j_r' \leq j_r$ and $j_l \equiv j_l' \equiv j_r' \equiv j_r \pmod d$. By the minimality of $q$ with $j_l \leq q \leq R_i$ and $t[q \dots R_i] = p[q-j_l \dots q-j_l+R_i-q]$, we have $q \leq j_l'$, and therefore the candidate determined by $q$ is
    \[
        j_q = j_l + \ceil*{\frac {q-j_l} d} d \leq j_l'.
    \]
    On the other hand, if $j_q \neq j_l'$, the existence of $j_l'$ leads immediately to that $t[j_q \dots j_q + m - 1]$ matches $p$, i.e. $j_q < j_l'$ is also an occurrence, which contradicts with the minimality of $j_l'$. As a result, we have $j_q = j_l'$. That is, our algorithm finds the leftmost occurrence in the block, if exists.

    {\vskip 3pt}

      \textbf{Complexity}. According to the above discussion, it is clear that $h_i$ can be computed with bounded-error in $O\left(\sqrt{m \log m} \right)$ queries.
    Thus, the entire problem can be solved by searching on bounded-error oracles (by Theorem \ref{thm-search-bounded-error-oracle}) and the query complexity is
    \[
    O\left(\sqrt{n/L} \sqrt{m \log m} \right) = O\left(\sqrt{n \log m} \right).
    \]

    Combining the above two cases, we conclude that there is a bounded-error quantum algorithm for pattern matching in $O\left(\sqrt{n \log m} + \sqrt {m \log^{3} m \log \log m}\right)$ queries.

\section{Proof of Exclusion Rule for LMSR} \label{appendix-ex-scr}

In this appendix, we present a proof of Lemma \ref{lemma-ex}. To this end, we first observe:

\begin{proposition} \label{prop-ex}
        Suppose $s \in \Sigma^n$ and $s[\operatorname{LMSR}(s) \dots \operatorname{LMSR}(s)+B-1] = aba$, where $\abs{a} \geq 1$, $\abs{b} \geq 0$, and $B = 2\abs{a}+\abs{b} \leq n/2$. For every $m > 0$ and $i \in [n]$, if $s[i \dots i+\abs{b}+\abs{a}-1] = ba$, then $$s[i \dots i+m-1] \leq s[i+\abs{b}+\abs{a} \dots i+\abs{b}+\abs{a}+m-1].$$
    \end{proposition}
    \begin{proof}
        We prove it by induction on $m$.

        \textbf{Basis}. For every $i \in [n]$ with $s[i \dots i+\abs{b}+\abs{a}-1] = ba$, we note that $s[i+\abs{b} \dots i+\abs{b}+B-1] = a s[i+\abs{b}+\abs{a} \dots i+\abs{b} + B-1]$. On the other hand, by the definition of LMSR, we have
        \[
            s[i+\abs{b} \dots i+\abs{b}+B-1] \geq s[\operatorname{LMSR}(s) \dots \operatorname{LMSR}(s)+B-1] = aba.
        \]
        Therefore, it holds that $s[i+\abs{b}+\abs{a} \dots i+\abs{b}-B-1] \geq ba = s[i \dots i+\abs{b}+\abs{a}-1]$. Immediately, we see that the proposition holds for $1 \leq m \leq \abs{b}+\abs{a}$.

        \textbf{Induction}. Assume that the proposition holds for $m' = k(\abs{b}+\abs{a})$ and $k \geq 1$, and we are going to prove it for the case $m' < m \leq (k+1)(\abs{b}+\abs{a})$. According to the induction hypothesis, we have
        \[
            s[i \dots i+m'-1] \leq s[i+\abs{b}+\abs{a} \dots i+\abs{b}+\abs{a}+m'-1]
        \]
        for every $0 \leq i < n$. Let us consider the following two cases:
        \begin{enumerate}
          \item $s[i \dots i+m'-1] < s[i+\abs{b}+\abs{a} \dots i+\abs{b}+\abs{a}+m'-1]$. In this case, it is trivial that $s[i \dots i+m-1] < s[i+\abs{b}+\abs{a} \dots i+\abs{b}+\abs{a}+m-1]$ for every $m > m'$.
          \item $s[i \dots i+m'-1] = s[i+\abs{b}+\abs{a} \dots i+\abs{b}+\abs{a}+m'-1]$. In this case, we have
          $s[i \dots i+(k+1)(\abs{b}+\abs{a})-1] = (ba)^{k+1}$, and $$s[i+\abs{b}+\abs{a} \dots i+(k+2)(\abs{b}+\abs{a})-1] = (ba)^k s[i+(k+1)(\abs{b}+\abs{a}) \dots i+(k+2)(\abs{b}+\abs{a})-1].$$ According to the induction hypothesis for $i' = i+m' = i+k(\abs{b}+\abs{a})$ (this can be derived from the index $i'' = i' \bmod n \in [n]$), we have:
          \[
            s[i+k(\abs{b}+\abs{a}) \dots i+(k+1)(\abs{b}+\abs{a})-1] \leq s[i+(k+1)(\abs{b}+\abs{a}) \dots i+(k+2)(\abs{b}+\abs{a})-1].
          \]
          Therefore, we obtain $s[i \dots i+(k+1)(\abs{b}+\abs{a})-1] \leq s[i+\abs{b}+\abs{a} \dots i+(k+2)(\abs{b}+\abs{a})-1]$, which means the proposition holds for $m = m' + \abs{b} + \abs{a}$. Immediately, we see that the proposition also holds for $m' < m \leq (k+1)(\abs{b}+\abs{a})$.
        \end{enumerate}
    \end{proof}

    Now we are ready to prove Lemma \ref{lemma-ex}.
    Let $\delta = j-i$, then we have $1 \leq \delta \leq B-1$. We consider the following two cases:

\begin{itemize}\item \textbf{Case 1}. $\delta > B/2$. In this case, $s[\operatorname{LMSR}(s) \dots \operatorname{LMSR}(s)+B-1] = aba$ for some strings $a$ and $b$, where $\abs{a} = B-\delta$ and $\abs{b} = 2\delta-B$. In order to prove that $\operatorname{LMSR}(s) \neq j$, it is sufficient to show that $s[i \dots i+n-1] \leq s[j \dots j+n-1]$. Note that
    \begin{align*}
        s[i \dots i+n-1] & = ababa s[j+B \dots j+n-\delta-1], \\
        s[j \dots j+n-1] & = aba s[j+B \dots j+n-1].
    \end{align*}
    We only need to show that $ba s[j+B \dots j+n-\delta-1] \leq s[j+B \dots j+n-1]$, that is, $s[i+B \dots i+n-1] \leq s[i+B+\delta \dots i+\delta+n-1]$, which can be immediately obtained from Proposition \ref{prop-ex} by letting $m \equiv n-B$ and $i \equiv i+B$.

\item \textbf{Case 2}. $\delta \leq B/2$. Let $t = s[\operatorname{LMSR}(s) \dots \operatorname{LMSR}(s)+B-1]$. Note that $t[i+\delta] = t[i]$ for every $i \in [B]$. We immediately see that $t$ has period $d = \gcd(B, \delta)$, that is, $t = a^k$, where $\abs{a} = d$ and $B = kd$ with $k \geq 2$. For convenience, we denote $\delta = ld$ for some $1 \leq l \leq k-1$. In order to prove that $\operatorname{LMSR}(s) \neq j$, it is sufficient to show that $s[i \dots i+n-1] \leq s[j \dots j+n-1]$. Note that
        \begin{align*}
            s[i \dots i+n-1] & = a^{k+m} s[j+B \dots j+n-\delta-1], \\
            s[j \dots j+n-1] & = a^k s[j+B \dots j+n-1].
        \end{align*}
    We only need to show that $a^m s[j+B \dots j+n-\delta-1] \leq s[j+B \dots j+n-1]$, i.e. $s[i+B \dots i+n-1] \leq s[i+B+\delta \dots i+\delta+n-1]$, which can be immediately obtained from Proposition \ref{prop-ex} by noting that $s[\operatorname{LMSR}(s) \dots \operatorname{LMSR}(s)+B-1] = a^l a^{k-2l} a^l$.\end{itemize}

\section{Worst-case Quantum Lower Bound} \label{appendix-a}

The worst-case quantum lower bound can be examined in a way different from that in Section \ref{worst-case-section}.
    Let us consider a special case where all strings are binary, that is, the alphabet is $\Sigma = \{0, 1\}$.
    A solution to the LMSR problem implies the existence of a $0$ character. That is, $s[i] = 0$ for some $i \in [n]$ if and only if $s[\operatorname{LMSR}(s)] = 0$. Therefore, the search problem can be reduced to LMSR. It is known that the search problem has a worst-case query complexity lower bound $\Omega\left(\sqrt n\right)$ for bounded-error quantum algorithms \cite{Ben97, Boy98, Zal99}, and $\Omega(n)$ for exact and zero-error quantum algorithms \cite{Bea01}. Consequently, we assert that the LMSR problem also has a worst-case quantum query complexity lower bound $\Omega\left(\sqrt n\right)$ for bounded-error quantum algorithms, and $\Omega(n)$ for exact and zero-error quantum algorithms.

\end{document}